\documentclass[12pt]{article}

\usepackage[utf8]{inputenc}
\usepackage[T1]{fontenc}
\usepackage{amsmath,amssymb,amsthm}
\usepackage{mathtools}
\usepackage[margin=1.25in]{geometry}
\usepackage{booktabs}
\usepackage{array}
\usepackage{enumitem}
\usepackage{hyperref}
\usepackage{xcolor}
\usepackage{parskip}
\usepackage{algorithm}
\usepackage{algpseudocode}
\usepackage{pgfplots}
\pgfplotsset{compat=1.18}
\usepackage{tikz}
\usepackage{caption}
\usepackage{subcaption}
\usepackage[numbers,sort&compress]{natbib}

\newtheorem{theorem}{Theorem}[section]
\newtheorem{proposition}[theorem]{Proposition}
\newtheorem{lemma}[theorem]{Lemma}
\newtheorem{corollary}[theorem]{Corollary}
\newtheorem{definition}[theorem]{Definition}
\newtheorem{mainthm}{Theorem}

\theoremstyle{remark}

\newtheorem{remarknum}[theorem]{Remark}

\newcommand{\Mon}{\mathfrak{M}}

\newcommand{\Z}{\mathbb{Z}}
\newcommand{\N}{\mathbb{N}}

\newcommand{\Zp}{\mathbb{Z}_p}
\newcommand{\abs}[1]{\lvert #1 \rvert}
\newcommand{\norm}[1]{\lVert #1 \rVert}
\newcommand{\floor}[1]{\lfloor #1 \rfloor}
\newcommand{\ceil}[1]{\lceil #1 \rceil}
\newcommand{\vp}{v_p}
\newcommand{\hp}{h_\alpha}
\newcommand{\bfpi}{\boldsymbol{\pi}}

\newcommand{\eps}{\varepsilon}
\newcommand{\F}{\mathbb{F}}
\newcommand{\Cc}{\mathcal{C}}
\newcommand{\Dig}{\mathsf{D}}
\newcommand{\Hsh}{H}             

\algrenewcommand\algorithmicrequire{\textbf{Input:}}
\algrenewcommand\algorithmicensure{\textbf{Output:}}

\hypersetup{
  colorlinks=true,
  linkcolor=blue!60!black,
  citecolor=green!50!black,
  urlcolor=blue!60!black
}

\usepackage{algorithm}      
\usepackage{algpseudocode}

\makeatletter
\newenvironment{breakablealgorithm}
  {\begin{center}
     \refstepcounter{algorithm}%
     \hrule height.8pt depth0pt \kern2pt%
     \renewcommand{\caption}[2][\relax]{%
       {\raggedright\textbf{\fname@algorithm~\thealgorithm}\ ##2\par}%
       \ifx\relax##1\relax
         \addcontentsline{loa}{algorithm}{\protect\numberline{\thealgorithm}##2}%
       \else
         \addcontentsline{loa}{algorithm}{\protect\numberline{\thealgorithm}##1}%
       \fi
       \kern2pt\hrule\kern2pt}%
  }
  {\kern2pt\hrule\relax
   \end{center}}
\makeatother

\title{\textbf{Algorithms, Complexity, and Entropy of the\\
Bernard--Letac Fair-Sampling Construction}}

\author{Claude Gravel\\[4pt]
\normalsize Department of Computer Science, Toronto Metropolitan University\\
\normalsize Toronto, Ontario, Canada\\[2pt]
\normalsize ORCID: \href{https://orcid.org/0000-0001-5275-4953}{0000-0001-5275-4953}}
\date{}

\begin{document}
\maketitle

\begin{abstract}
Bernard and Letac (1971) introduced a method for uniform random sampling among $m$ outcomes from an unknown biased source of independent and identically distributed symbols. The process terminates when the multinomial coefficient of the cumulative symbol counts equals zero modulo $m$. This study extends the computational and information-theoretic analysis of their construction by presenting five algorithms with formal correctness guarantees and comprehensive complexity analyses. For prime $m=p$, the Bernard--Letac framework is analyzed in greater detail. The R\'enyi entropies of the source yield an exact product formula for the expected number of draws. A first-order approximation consistently overestimates this value, and the entropy lower bound is never attained. As $p$, treated as a continuous parameter, approaches 1, the expected cost converges to a constant greater than 1, determined by the entire source distribution. Furthermore, for every prime modulus $p$ and every finite alphabet $I$, an explicit automaton with $p + |I| + 3$ states computes the mod-$p$ first-passage kernel of the walk from the base-$p$ digits of its arguments, reducing the fair assignment cost from quadratic to nearly linear.
\end{abstract}



\section{Introduction and Background}
\label{sec:intro}
A central problem in probability and computation concerns the following question. Given access to a biased source of randomness, specifically a device producing independent and identically distributed (i.i.d.) draws from an unknown distribution $\bfpi$ over a finite alphabet, can one generate a fair outcome that is uniform over $m$ possibilities using a finite, albeit random, number of draws? The procedure must succeed for every admissible $\bfpi$, without any prior knowledge of the underlying bias.

The earliest solution to this problem was proposed by von Neumann~\cite{vonneumann}, who in 1951 demonstrated that a fair coin flip can be simulated using a biased coin through a simple rejection scheme. Letting $\pi_0$ and $\pi_1$ denote the probabilities of the two outcomes, each trial of this scheme consumes a pair of draws and succeeds with probability $2\pi_0\pi_1$. Thus, the expected cost is $1/(2\pi_0\pi_1)$ pairs, or equivalently, $1/(\pi_0\pi_1)$ draws. In 1970, Hoeffding and Simons~\cite{hoeffding} revisited the problem, developing improved procedures for the binary case that achieve a smaller expected stopping time and characterizing the class of all valid stopping rules. Notably, they introduced a walk on $\N^2$ tracked by the parity of the binomial coefficient $\binom{t}{k}$, where $t$ is the number of draws made so far and $k$ is the number of those equal to one fixed outcome. This device served as a direct precursor to the Bernard--Letac construction.

In 1971, Bernard and Letac~\cite{bernardletac} extended both previous works in two significant directions. They allowed the alphabet $I$ to be arbitrary, whether finite or countable, and aimed to produce an $m$-way fair partition rather than only a fair binary split. Their principal contribution is a construction based on the set $H_m = \{x \in \Mon : c(x) \equiv 0 \pmod{m}\}$, where $\Mon$ denotes the free abelian monoid on $I$ and $c(x)$ represents the multinomial coefficient. The random walk defined by the cumulative draw counts terminates upon first reaching $H_m$. The authors demonstrate that the paths to any stopping state can be partitioned into exactly $m$ equiprobable classes, regardless of the bias $\bfpi$.

The results of the original paper are primarily combinatorial and number-theoretic, with no consideration of computational or information-theoretic perspectives. The present work addresses these gaps by (1) extracting from the original theorems a family of five distinct algorithms (Section~\ref{sec:algorithms}), each presented formally with correctness guarantees, complexity analyses, and explicit pseudocode; (2) deriving new and precise results regarding the expected complexity of the central fair-sampling algorithm, expressed in terms of R\'enyi and Shannon entropy (Section~\ref{sec:complexity}); and (3) establishing a transfer theorem for the mod-$p$ first-passage kernel of the walk, for every prime $p$ and every finite alphabet (Section~\ref{sec:transfer}). Additionally, worked examples (Section~\ref{sec:examples}) and numerical illustrations (Section~\ref{sec:numerical}) are provided to concretely demonstrate the construction.

\subsection*{Main results}

The main results of the paper are the four theorems below. Each is proved in Sections~\ref{sec:complexity} and~\ref{sec:transfer}, where the statements are repeated with complete detail. The remaining notation is fixed later in this section.

To state the results, let $\bfpi = (\pi_i)_{i \in I}$ be a probability distribution on a finite or countable alphabet~$I$, non-degenerate in the sense that $\pi_i < 1$ for every~$i$, and let $T$ denote the number of draws the Bernard--Letac sampler consumes when run with modulus~$m$, that is, the first time the walk of cumulative symbol counts enters~$H_m$. Write $\Hsh(\bfpi) = -\sum_i \pi_i \log \pi_i$ for the Shannon entropy of the source and $H_\alpha(\bfpi) = (1 - \alpha)^{-1} \log \sum_i \pi_i^{\alpha}$ for its R\'enyi entropy of order~$\alpha \neq 1$; all logarithms are natural.

The first main theorem gives the entropy form of the expected stopping time.

\begin{mainthm}
\label{mainthm:A}
Let $m = p$ be a prime and set $U_n(\bfpi) = \sum_{i \in I} \pi_i^{p^n}$ for $n \geq 1$, so that $U_1(\bfpi) = e^{(1-p)\,H_p(\bfpi)}$. Then
\[
 E(T) \;=\; p \prod_{n=1}^{\infty} \frac{1 - (U_n(\bfpi))^p}{1 - U_n(\bfpi)}\,,
\]
and the first-order approximation
\[
 E(T) \;\approx\; \frac{p}{1 - e^{(1-p)\,H_p(\bfpi)}}
\]
always strictly overestimates, in the sense that $E(T) < p/(1 - U_1(\bfpi))$ for every non-degenerate~$\bfpi$.
\end{mainthm}

The exact product form is due to Bernard and Letac (Corollaire~4 of~\cite{bernardletac}, restated as Theorem~\ref{thm:ET}); the entropy reading, the first-order approximation, and the strict inequality are established in Theorem~\ref{thm:ET-renyi}.

The second main theorem identifies the Shannon limit of the expected cost.

\begin{mainthm}
\label{mainthm:B}
Treat $p > 1$ as a continuous parameter in the product of Theorem~\ref{mainthm:A}. Then $E(T)$ has a finite limit as $p \to 1^+$, namely
\[
 \lim_{p \to 1^+} E(T) \;=\; C(\bfpi) \;=\; \exp\!\left(\int_0^\infty \frac{-\,\varphi(s)\,\log \varphi(s)}{1 - \varphi(s)}\,ds\right),
 \qquad \varphi(s) = \sum_{i \in I} \pi_i^{\,e^s}.
\]
The constant $C(\bfpi)$ is strictly greater than~$1$, depends on the full distribution $\bfpi$ rather than on its Shannon entropy alone, and is not in general equal to $1/\Hsh(\bfpi)$. For the symmetric binary source one has $C(\bfpi) = 2.8457\ldots$, whereas $1/\Hsh(\bfpi) = 1.4427\ldots$.
\end{mainthm}

This is proved as Theorem~\ref{thm:shannon-limit}.

The third main theorem gives the information-theoretic lower bound and describes the efficiency of the sampler.

\begin{mainthm}
\label{mainthm:C}
Every procedure that observes the source at a stopping time~$T$ and returns an exactly uniform outcome among $m$ alternatives for every non-degenerate~$\bfpi$ satisfies
\[
 E(T) \;\geq\; \frac{\log m}{\Hsh(\bfpi)}\,.
\]
The inequality is never attained. For the Bernard--Letac sampler with prime modulus~$p$, define the efficiency $\eta(p) = \log p/(E(T)\,\Hsh(\bfpi))$. Then $\eta(p)$ tends to~$0$ both as $p \to 1^+$ and as $p \to \infty$, and it attains its maximum at an interior point $p^\star(\bfpi) \in (1, \infty)$.
\end{mainthm}

The lower bound is proved as Theorem~\ref{thm:lower-bound} by a Wald-type entropy identity, its strictness is Proposition~\ref{prop:lower-bound-strict}, and the behavior of the efficiency is Corollary~\ref{cor:efficiency}.

The fourth main theorem states, in informal form, the transfer theorem for the first-passage kernel.

\begin{mainthm}
\label{mainthm:D}
Let $m = p$ be a prime and let the alphabet $I$ be finite with $n = |I| \geq 2$ letters. For a stopping state $z \in H_p$ and a lattice point $v \leq z$, let $\varphi(v,z)$ count the monotone lattice paths from $v$ to $z$ whose points other than $z$ avoid~$H_p$. An explicit automaton with exactly $p + n + 3$ states reads the base-$p$ digits of the $2n$ coordinates of $(v, z)$, one digit column per step from the least significant position upward, and computes $\varphi(v,z)$ modulo~$p$; the automaton is minimal, and the kernels attached to its states span a space of dimension exactly $n + 2$ over $\F_p$. Consequently the exact fair group assignment of the sampler is carried out in $O(n\,T \log_p T)$ automaton transitions, which for $m = 2$ on the binary alphabet is $O(T \log T)$ bit operations, in place of the $O(T^2)$ bit operations of the direct dynamic program (Proposition~\ref{prop:per-step}) run modulo~$2$.
\end{mainthm}

The automaton is constructed in Section~\ref{sec:transfer}; the precise statements are Theorems~\ref{thm:transfer} and~\ref{thm:dimension}, the seven-state case $p = 2$ is Corollary~\ref{cor:binary-table}, and the complexity consequence is Corollary~\ref{cor:P1-binary}.

\subsection*{Related work}

In addition to the direct progression from Von Neumann through Hoeffding--Simons to Bernard--Letac, several subsequent developments contextualize the present results. Elias~\cite{elias} developed an asymptotically optimal block procedure whose efficiency approaches~$1$ as the block length increases, so that its output rate approaches $\Hsh(\bfpi)/\log m$ outcomes per draw; the price of this optimality is unbounded delay. In contrast, the Bernard--Letac algorithm operates online and sequentially, prioritizing immediacy over efficiency. Peres~\cite{peres} demonstrated that iterating Von Neumann's procedure yields, in the limit, a number of unbiased bits arbitrarily close to the entropy bound, so that its efficiency tends to~$1$. The relationship between iterated extraction and the $p$-adic limit $L(x)$ of Algorithm~\ref{alg:padic} remains an open problem of interest. The Knuth--Yao framework~\cite{knuthyao} establishes the minimum expected number of fair coin flips required to simulate any distribution, leading to the lower bound $E(T) \geq \log m / \Hsh(\bfpi)$ stated in Theorem~\ref{thm:lower-bound}. The Bernard--Letac construction extends this framework to arbitrary source alphabets and moduli. Lumbroso~\cite{lumbroso} provided an optimal implementation of the Knuth--Yao discrete distribution generating (DDG) tree for uniform distributions (the ``Fast Dice Roller''), which recycles entropy from rejections and serves as the natural benchmark for competing uniform samplers. The composite-modulus comparison in Section~\ref{sec:benchmark} employs the Von Neumann plus Lumbroso method as the strongest baseline.

From a technical perspective, the combinatorial foundation of the paper, specifically the congruences of multinomial coefficients modulo prime powers, is explored in detail by Fray~\cite{fray}, Granville~\cite{granville}, and Deutsch--Sagan~\cite{deutschsagan}. The $p$-adic convergence result of Algorithm~\ref{alg:padic} is connected to the theory of $p$-adic $\Gamma$-functions, as discussed by Morita~\cite{morita}; standard background on $p$-adic numbers and valuation theory is given by Bachman~\cite{bachman}. The emergence of R\'enyi entropy as the natural complexity parameter for algorithms involving $q$-th power sums is well documented; see the survey by Acz\'el and Dar\'oczy~\cite{aczel}. Theorem~\ref{thm:ET-renyi} appears to represent the first instance of this connection within the context of fair sampling.

\subsection*{Notation and background}

Throughout, let $I$ denote a finite or countable alphabet and $\Mon$ the free abelian monoid generated by $I$, identified with the set of finitely supported sequences of non-negative integers $(x_i)_{i \in I}$. The notation $\abs{x} = \sum_i x_i$ is used, and the \emph{multinomial coefficient} is defined as follows.
\[
 c(x) = \frac{\abs{x}!}{\prod_i x_i!}, \qquad x \in \Mon.
\]
We fix a prime~$p$. For a non-negative integer $n$, we write $n = \sum_{\alpha \geq 0} h_\alpha(n)\,p^\alpha$ for its base-$p$ expansion with $0 \leq h_\alpha(n) < p$, and define
\[
 \vp(n) = \inf\{\alpha : h_\alpha(n) > 0\}, \quad
 d(n) = \sup\{\alpha : h_\alpha(n) > 0\}, \quad
 \norm{n} = \sum_{\alpha \geq 0} h_\alpha(n),
\]
with $\vp(0) = +\infty$ and $d(0) = -\infty$. For $x \in \Mon$ one extends these componentwise: $h_\alpha(x) = (h_\alpha(x_i))_{i\in I}$, $\vp(x) = \inf_i \vp(x_i)$, $d(x) = \sup_i d(x_i)$, and $\norm{x} = \norm{\abs{x}}$. The monoid $\Mon$ carries the componentwise partial order, in which $x \leq y$ means $x_i \leq y_i$ for every $i \in I$.

A probability distribution on~$I$ is a vector $\bfpi = (\pi_i)_{i \in I}$ with $\pi_i \geq 0$, $\sum_i \pi_i = 1$, and $\pi_i < 1$ for all~$i$ (the non-degenerate condition). We write $\bfpi^x = \prod_i \pi_i^{x_i}$ for the monomial associated with $x \in \Mon$.

A comprehensive development of the Bernard--Letac theory, including all proofs, is provided in~\cite{bernardletac}. For convenience, the three principal results underpinning the algorithms of Section~\ref{sec:algorithms} are restated here without proof.

Let $(I, \bfpi)$ be a non-degenerate source. Form the product space
\begin{displaymath}
(\Omega, \mathcal{F}, P) = \prod_{t=1}^{\infty}(I, \mathcal{P}(I), \bfpi),
\end{displaymath}
with canonical projections $a_t : \Omega \to I$ and natural filtration $(\mathcal{F}_t)_{t \geq 0}$. Define the $\Mon$-valued random walk by $S_0 = 0$ and $S_t = \sum_{s=1}^t \eps^{a_s}$, where $\eps^i \in \Mon$ is the $i$-th standard unit vector, the indicator vector of symbol~$i$. Then $S_t = (n_i^{(t)})_{i\in I}$ records the
empirical count of each symbol after $t$ draws, and
\[
 P(S_t = x) = c(x)\,\bfpi^x \quad \text{for all } x \in \Mon \text{ with } \abs{x} = t.
\]

For a set $E \subset \Mon$ define the hitting time $T(E) = \inf\{t \geq 1 : S_t \in E\}$. Let $\mathfrak{A}_m$ be the class of stopping times~$T$ such that there exists an $\mathcal{F}_T$-measurable partition $(A_1, \ldots, A_m)$ of $\Omega$ with $P(A_j) = 1/m$ for all~$j$ and all admissible~$\bfpi$.

A central quantity is $m_E(x)$, the number of paths from $0$ to $x$ in $\Mon$ that do not touch $E \setminus \{x\}$.

The following fundamental theorem is due to Bernard and Letac; see \S2 of~\cite{bernardletac}.

\begin{theorem}
\label{thm:BL-fundamental}
For $E \subset \Mon$, the following are equivalent:
\begin{enumerate}[label=\emph{(\alph*)}]
 \item $m_E(x) \equiv 0 \pmod{m}$ for all $x \in E$;
 \item $c(x) \equiv 0 \pmod{m}$ for all $x \in E$;
 \item $m_E(x) \equiv c(x) \pmod{m}$ for all $x \in \Mon$.
\end{enumerate}
\end{theorem}

The next theorem, also from \S2 of~\cite{bernardletac}, identifies the stopping set used throughout this paper.

\begin{theorem}
\label{thm:BL-Hm}
Let $H_m = \{x \in \Mon : c(x) \equiv 0 \pmod{m}\}$. Then $T(H_m) \in \mathfrak{A}_m$, i.e., the stopping rule based on $H_m$ produces a fair $m$-way partition.
\end{theorem}

The following theorem of Kummer~\cite{kummer,dickson} is the key number-theoretic tool. (See \S1 of~\cite{bernardletac} for the original proof.)

\begin{theorem}
\label{thm:kummer}
For $x \in \Mon$, $\vp(c(x)) = \sum_{\alpha \geq 0} q_\alpha(x)$, where $q_\alpha(x) \geq 0$ is the carry in column~$\alpha$ when adding the base-$p$ representations of $(x_i)_{i \in I}$.
\end{theorem}

An immediate corollary is the multinomial analogue of Lucas' theorem, namely that $c(x) \not\equiv 0 \pmod{p}$ if and only if $\sum_i h_\alpha(x_i) < p$ for every digit position~$\alpha$.

\section{The Algorithms}
\label{sec:algorithms}

Five distinct algorithms are extracted from the results of~\cite{bernardletac}, each accompanied by a formal correctness statement and a complexity guarantee. Theorem~\ref{thm:kummer} directly provides an efficient procedure for computing $\vp(c(x))$.

\begin{proposition}
\label{prop:alg1-complexity}
Algorithm~\ref{alg:kummer} correctly computes $\vp(c(x))$ in
$O\bigl(|I| \cdot (\log_p(\max_i x_i) + \log_p |I|)\bigr)$ time, with the convention $\log_p 0 = 0$, so that $x = \mathbf{0}$ is a valid input.
\end{proposition}

\begin{proof}
Correctness follows from Theorem~\ref{thm:kummer}, provided the loop runs until the carry is exhausted. At each digit position~$\alpha$, one performs a single pass over the $|I|$ components to compute $q_\alpha$. There are at most $\floor{\log_p(\max_i x_i)} + 1$ positions at which some $x_i$ has a non-zero digit. Beyond them the column sum is the carry alone; since $\text{col\_sum} \leq \text{carry} + |I|(p-1)$ the invariant $\text{carry} \leq |I|$ is maintained throughout, so the carry is extinguished within a further $\ceil{\log_p |I|}$ iterations. The total work is therefore $O\bigl(|I| \cdot (\log_p(\max_i x_i) + \log_p |I|)\bigr)$. \end{proof}

The incremental variant, which updates $\vp(c(S_t))$ as $S_t$ changes by $\eps^{a_t}$ at each step, requires $O(\log_p t)$ time per step. This is because only one coordinate changes and carries propagate through at most $O(\log_p t)$ positions.

\begin{algorithm}[H]
\caption{$p$-adic Valuation of $c(x)$ via Kummer Carries}
\label{alg:kummer}
\begin{algorithmic}[1]
\Require $x = (x_i)_{i \in I} \in \Mon$, prime $p$
\Ensure $\vp(c(x))$
\State $\text{val} \gets 0$, $\text{carry} \gets 0$, $\alpha \gets 0$
\State $D \gets \floor{\log_p(\max_i x_i)} + 1$
 \Comment{with the convention $\log_p 0 = 0$}
\While{$\alpha \leq D$ \textbf{ or } $\text{carry} > 0$}
 \Comment{trailing carries survive past $D$ when $|I| > p$}
 \State $\text{col\_sum} \gets \text{carry} + \sum_{i \in I} h_\alpha(x_i)$
 \Comment{digit sum at position $\alpha$ plus incoming carry}
 \State $q_\alpha \gets \floor{\text{col\_sum}/p}$
 \Comment{carry into position $\alpha+1$}
 \State $\text{val} \gets \text{val} + q_\alpha$
 \State $\text{carry} \gets q_\alpha$, $\alpha \gets \alpha + 1$
\EndWhile
\State \Return $\text{val}$
\end{algorithmic}
\end{algorithm}

In addition to its function within the sampler, the carry structure computed by Algorithm~\ref{alg:kummer} is relevant to a corollary of Section 5 in~\cite{bernardletac}, which provides an upper bound on the number of zero-one matrices with specified margins. Improving this bound remains an open problem, designated as Open Problem~P7 (Section~\ref{sec:open}).

Theorems~\ref{thm:BL-fundamental} and~\ref{thm:BL-Hm} provide the theoretical foundation for the following procedure.

\begin{theorem}
\label{thm:alg2-correct}
Algorithm~\ref{alg:fairsample} halts almost surely for every non-degenerate $\bfpi$, and produces a random variable~$J$ with $P(J = j) = 1/m$ for all $j \in \{1,\ldots,m\}$, independently of~$\bfpi$.
\end{theorem}

\begin{proof}
Almost sure termination holds since $\pi_i < 1$ for all~$i$, the walk $S_t$ is multi-dimensional, and the probability of never hitting $H_m$ is zero~\cite[Theorem~4]{bernardletac}, whose proof proceeds via the recurrence of the induced walk on $(\Z/m\Z)^2$ (see~\cite{feller} for the underlying recurrence criteria).

For the fairness of the output, we have that since $T = T(H_m)$ is a first hitting time, the trajectories that can be observed at~$x$ are precisely the $m_E(x)$ paths from $\mathbf{0}$ to~$x$ avoiding $H_m \setminus \{x\}$. A path that meets $H_m$ earlier stops earlier and never reaches~$x$ at time~$T$. Each such path has probability $\bfpi^x$, so conditionally on $\{S_T = x\}$ for $x \in H_m$, the observed path is uniform over these $m_E(x)$ paths. Since $c(y) \equiv 0 \pmod{m}$ for every $y \in H_m$ by definition, the implication (b)\,$\Rightarrow$\,(a) of Theorem~\ref{thm:BL-fundamental} gives $m \mid m_E(x)$. Partitioning the avoiding paths into $m$ groups of size $m_E(x)/m$ therefore yields $P(A_j \cap \{S_T = x\}) = \tfrac{m_E(x)}{m}\,\bfpi^x = P(S_T = x)/m$; summing over stopping states gives $P(A_j) = 1/m$ for every admissible $\bfpi$, because the counts $m_E(x)$ are combinatorial quantities independent of the weights.
\end{proof}

The demonstrated distribution-free fairness depends on the stopping set $H_m$ being specified prior to any sampling. Whether a data-adaptive stopping boundary can decrease $E(T)$ while preserving exact fairness remains unresolved and is identified as Open Problem~P4 (Section~\ref{sec:open}).

\begin{algorithm}[H]
\caption{Fair $m$-way Sampling from a Biased Source}
\label{alg:fairsample}
\begin{algorithmic}[1]
\Require Alphabet $I$, modulus $m$, access to i.i.d.\ draws from unknown $\bfpi$
\Ensure Outcome $J \in \{1,\ldots,m\}$ with $P(J = j) = 1/m$
\State $S \gets \mathbf{0} \in \Mon$; $\text{path} \gets []$
\Repeat
 \State Draw $a \sim \bfpi$; append $a$ to path; $S \gets S + \eps^a$
 \State For each prime power $p_j^{a_j} \,\|\, m$, compute $v_{p_j}(c(S))$ via Algorithm~\ref{alg:kummer}
 \Comment{$\omega(m)$ valuations per draw, as in Proposition~\ref{prop:per-step}}
\Until{$v_{p_j}(c(S)) \geq a_j$ for every $j$, i.e.\ $c(S) \equiv 0 \pmod{m}$}
\For{$v$ in the box $[\mathbf{0}, S] = \{v \in \Mon : v \leq S\}$, in decreasing order of $\abs{v}$}
 \State $A(v) \gets \mathbf{1}[v = S]$ \textbf{ if } $v \in H_m$, \textbf{ else } $\sum_{i \in I} A(v + \eps^i)$
 \Comment{$A \equiv 0$ outside $[\mathbf{0}, S]$}
\EndFor
 \Comment{$A(v)$ counts the paths $v \to S$ avoiding $H_m$; $A(\mathbf{0}) = m_E(S)$}
\State $r \gets \sum_{t=1}^{T} \sum_{b < a_t} A(S_{t-1} + \eps^b)$
 \Comment{rank of observed path among the $m_E(S)$ paths avoiding $H_m$}
\State \Return $J \gets \floor{r \cdot m / A(\mathbf{0})} + 1$
 \Comment{which group of size $m_E(S)/m$ the path falls in}
\end{algorithmic}
\end{algorithm}

\begin{remarknum}[The partition must respect first hitting]
\label{rem:avoiding-pitfall}
One might consider partitioning all $c(x)$ paths to $x$ into $m$ groups. This approach would create lexicographic blocks of size $c(x)/m$, requiring only $O(T \cdot |I|)$ arithmetic operations. However, this method is not equiprobable because the $m_E(x)$ observable paths are not necessarily distributed evenly among the blocks. For example, when $|I| = 2$ and $m = 2$, the state $x = (3,1) \in H_2$ has $c(x) = 4$ but only two observable paths, $0001$ and $0010$. The remaining two paths pass through $(1,1) \in H_2$ and terminate there. Both observable paths are contained within the first lexicographic block, while both observable paths to the mirror state $(1,3)$ are in the second block. The two contributions, $2\pi_0^3\pi_1$ and $2\pi_0\pi_1^3$, cancel only for the unbiased source. Exact enumeration with $\bfpi = (0.7, 0.3)$ yields $P(J = 1) \approx 0.70$, indicating that the block rule essentially reproduces the bias it was intended to eliminate.
\end{remarknum}

The following result of~\cite{bernardletac}, known as the Triangle Theorem, provides a fast membership test. (See \S4 of~\cite{bernardletac} for the full proof.)

\begin{theorem}
\label{thm:triangle}
Let $\eta = (\eta_i)_{i\in I}$ with $\eta_i = \pm 1$. If $\vp(x_i) > d(y)$ for all $i$ and $|\eta y| \geq 0$, then $\vp(c(x + \eta y)) \geq \vp(c(x))$. In particular, if $x \in H_{p^n}$ then $x + \eta y \in H_{p^n}$ for all such perturbations~$\eta y$.
\end{theorem}

The membership oracle built on this test has the following guarantees.

\begin{proposition}
\label{prop:alg3}
Algorithm~\ref{alg:membership} decides $x \in H_{p^n}$ in $O(|I| \cdot \log_p(\max_i x_i))$ time. Suppose moreover that a set $G \subset H_{p^n}$ of \emph{$p$-adically deep} generators, meaning points $g$ with $\vp(g_i) \geq n$ for every~$i$, has been precomputed. Such generators exist, for example the points $g = p^n y$ with $y \in H_{p^n}$; these lie in $H_{p^n}$ because $\vp(c(p^n y)) = \vp(c(y)) \geq n$, the base-$p$ carries of Theorem~\ref{thm:kummer} being invariant under the digit shift $y \mapsto p^n y$. Then membership is decided in $O(|I|)$ time for every $x$ obtained from some $g \in G$ by a perturbation $x - g$ that is shallow in the sense of Theorem~\ref{thm:triangle}.
\end{proposition}

\begin{algorithm}[H]
\caption{Membership Oracle for $H_{p^n}$ with Triangle Acceleration}
\label{alg:membership}
\begin{algorithmic}[1]
\Require Query $x \in \Mon$, prime $p$, exponent $n$, generator set $G$
\Ensure Boolean: $x \in H_{p^n}$?
\For{$g \in G$}
 \If{$x - g \in \Mon$ \textbf{and} $\vp(g_i) > d(x - g)$ for all $i$}
 \State \Return \textbf{true}
 \Comment{Triangle Theorem: $g \in H_{p^n}$ implies $x \in H_{p^n}$}
 \EndIf
\EndFor
\State Compute $\vp(c(x))$ via Algorithm~\ref{alg:kummer}
\State \Return $\vp(c(x)) \geq n$
\end{algorithmic}
\end{algorithm}

The following periodicity theorem of~\cite{bernardletac} makes a precomputed lookup table possible. (See \S5 of~\cite{bernardletac} for the proof.)

\begin{theorem}
\label{thm:periodicity}
For fixed $x \in \Mon$, $i \in I$ with $x_i = 0$, and $n \geq 0$:
\[
 \vp\!\bigl(c(x + t\eps^i) - c(x + (t + p^n)\eps^i)\bigr) \;\geq\; n - d(x)
 \quad \text{for all } t \geq 0.
\]
\end{theorem}

Theorem~\ref{thm:periodicity} yields the following guarantee for the lookup table of Algorithm~\ref{alg:lookup}.

\begin{proposition}
\label{prop:alg4}
Algorithm~\ref{alg:lookup} precomputes a table of $c(x + t\eps^i) \bmod p^n$ for $t = 0, 1, \ldots, p^{\,n+d(x)} - 1$ in $O(p^{\,n+d(x)})$ time and space, after which any query $c(x + t\eps^i) \bmod p^n$ is answered in $O(1)$. When $d(x) = 0$ the period is $p^n$, as one would expect.
\end{proposition}

Theorem~\ref{thm:periodicity} controls $c(x + t\eps^i)$ only modulo $p^{\,n - d(x)}$ at period $p^n$, so a table of length $p^n$ is in general too short. Invoking the theorem at exponent $n + d(x)$ restores a guaranteed period modulo $p^n$. For $x = (0,3)$, $p = 2$, $n = 2$ (so $d(x) = 1$) the length-$4$ table is $[1,0,2,0]$, yet $c(4,3) = 35 \equiv 3 \pmod 4$ while $T[4 \bmod 4] = 1$; the length-$8$ table $[1,0,2,0,3,0,0,0]$ is periodic.

\begin{algorithm}[H]
\caption{Periodic Lookup Table for $c(x + t\eps^i) \bmod p^n$}
\label{alg:lookup}
\begin{algorithmic}[1]
\Require Base $x \in \Mon$ with $x_i = 0$, index $i$, prime $p$, exponent $n$
\Ensure Precomputed table $T[0 \ldots P - 1]$, $P = p^{\,n+d(x)}$
\State $P \gets p^{\,n + d(x)}$
 \Comment{Theorem~\ref{thm:periodicity} at exponent $n + d(x)$: period $P$ modulo $p^n$}
\For{$t = 0, 1, \ldots, P - 1$}
 \State $T[t] \gets c(x + t\,\eps^i) \bmod p^n$
\EndFor
\State \Return $T$ \Comment{Query: $c(x + t\eps^i) \bmod p^n = T[t \bmod P]$}
\end{algorithmic}
\end{algorithm}

The following theorem of~\cite{bernardletac} establishes $p$-adic convergence of the multinomial coefficients. (See \S6 of~\cite{bernardletac} for the proof.)

\begin{theorem}
\label{thm:padic}
For $x \in \Mon$, $x \neq 0$:
\[
 \vp\!\bigl(c(p^n x) - c(p^{n-1}x)\bigr) \;\geq\; n + \vp(x) + \vp(c(x)).
\]
Consequently, $c(p^n x)$ converges in $\Zp$ to a $p$-adic integer $L(x)$.
\end{theorem}

Theorem~\ref{thm:padic} yields the following guarantee for Algorithm~\ref{alg:padic}.

\begin{proposition}
\label{prop:alg5}
Algorithm~\ref{alg:padic} computes $L(x) \bmod p^N$ in $O(N)$ iterations, with each iteration requiring $O(|I| \log p)$ arithmetic operations.
\end{proposition}

For an index set $K$ we write $\Mon_K$ for the free abelian monoid on $K$, so that an element $z \in \Mon_{I \times J}$ is a finitely supported matrix, and its margins are its row-sum vector in $\Mon_I$ and its column-sum vector in $\Mon_J$. The limit satisfies a product identity (Th\'eor\`eme~10 of~\cite{bernardletac}). If $x \in \Mon_I$ and $y \in \Mon_J$ (for a second index set $J$) satisfy $|x| = |y|$ and $\inf(\vp(x), \vp(y)) = 0$, then
\[
 L(x)\,L(y) \;=\; \sum_z L(z),
\]
where $z$ ranges over the elements of $\Mon_{I \times J}$ whose margins are $p^k x$ and $p^k y$ for some $k \geq 0$ and which satisfy $\vp(z) = 0$; the series converges $p$-adically. The restriction $\inf(\vp(x), \vp(y)) = 0$ is only a normalization, since $L(px) = L(x)$.

Determining whether the limit $L(x)$ possesses an information-theoretic interpretation, such as through the entropy of a distribution associated with $x$, constitutes Open Problem~P5 (Section~\ref{sec:open}).

\begin{algorithm}[H]
\caption{$p$-adic Limit $L(x) = \lim_{n\to\infty} c(p^n x) \in \Zp$}
\label{alg:padic}
\begin{algorithmic}[1]
\Require $x \in \Mon$, prime $p$, target precision $N$
\Ensure $L(x) \bmod p^N$
\State $\text{prev} \gets c(x) \bmod p^N$
\For{$k = 1, 2, \ldots, N$}
 \State $\text{curr} \gets c(p^k x) \bmod p^N$
 \If{$\text{curr} \equiv \text{prev} \pmod{p^N}$}
 \State \Return $\text{curr}$
 \Comment{Convergence detected at precision $p^N$}
 \EndIf
 \State $\text{prev} \gets \text{curr}$
\EndFor
\State \Return $\text{curr}$
\end{algorithmic}
\end{algorithm}

\subsection{Two worked examples}
\label{sec:examples}

\subsubsection*{Example 1: prime modulus}

Let $I = \{0,1\}$ and $m = 5$ (prime). The random walk $S_t = (k, t-k)$ records draws of each symbol; $c(k,t-k) = \binom{t}{k}$. The stopping condition is $\binom{t}{k} \equiv 0 \pmod{5}$. The \emph{level}~$t$ of the walk is the set of states $x$ with $\abs{x} = t$; since $\abs{S_t} = t$, the walk occupies level~$t$ exactly at time~$t$.

\medskip
For the first stopping level, we have according to Lucas' theorem, $\binom{t}{k} \equiv 0 \pmod{5}$ first occurs at $t = 5$ for $k \in \{1,2,3,4\}$. If all five draws are identical, the random walk does not stop at this stage; otherwise, it stops at $t = 5$.

The state $S_5 = (2,3)$ is reached when two symbols $0$ and three symbols $1$ have appeared. In this case, $c(2,3) = \binom{5}{2} = 10 \equiv 0 \pmod{5}$. Since no state with $t < 5$ lies in $H_5$, no path can stop before $t = 5$. Therefore, $m_E((2,3)) = c((2,3)) = 10$, and every path to $(2,3)$ is observable. This property is unique to the first stopping level; see Example~2 and Remark~\ref{rem:avoiding-pitfall}.

\medskip
Each path to $(2,3)$ corresponds to a sequence of five draws containing exactly two 0s and three 1s. There are $\binom{5}{2} = 10$ such paths, which can be listed lexicographically:

\begin{center}
\begin{tabular}{cl|cl}
\toprule
\textbf{Path} & \textbf{Sequence} & \textbf{Path} & \textbf{Sequence} \\
\midrule
1 & $0\ 0\ 1\ 1\ 1$ & 6 & $1\ 0\ 1\ 0\ 1$ \\
2 & $0\ 1\ 0\ 1\ 1$ & 7 & $1\ 0\ 1\ 1\ 0$ \\
3 & $0\ 1\ 1\ 0\ 1$ & 8 & $1\ 1\ 0\ 0\ 1$ \\
4 & $0\ 1\ 1\ 1\ 0$ & 9 & $1\ 1\ 0\ 1\ 0$ \\
5 & $1\ 0\ 0\ 1\ 1$ & 10 & $1\ 1\ 1\ 0\ 0$ \\
\bottomrule
\end{tabular}
\end{center}

Each path has probability $\pi_0^2\pi_1^3$. Since $10/5 = 2$, assign two consecutive paths to each group:

\begin{center}
\begin{tabular}{ccc}
\toprule
\textbf{Group $A_j$} & \textbf{Paths} & \textbf{Probability} \\
\midrule
$A_1$ & 1, 2 & $2\pi_0^2\pi_1^3$ \\
$A_2$ & 3, 4 & $2\pi_0^2\pi_1^3$ \\
$A_3$ & 5, 6 & $2\pi_0^2\pi_1^3$ \\
$A_4$ & 7, 8 & $2\pi_0^2\pi_1^3$ \\
$A_5$ & 9, 10 & $2\pi_0^2\pi_1^3$ \\
\bottomrule
\end{tabular}
\end{center}

Each group accounts for exactly $\tfrac{1}{5}$ of the total probability mass at the stopping event $\{S_5 = (2,3)\}$, regardless of $\pi_0, \pi_1$. For instance, the observed sequence $1,0,1,0,1$ is Path~6, yielding outcome~$A_3$.

\subsubsection*{Example 2: composite modulus}

Let $I = \{0,1,2\}$ and $m = 6 = 2 \times 3$ (composite). The state $S_t = (k_1,k_2,k_3)$ records draw counts; $c(k_1,k_2,k_3) = (k_1+k_2+k_3)!/(k_1!\,k_2!\,k_3!)$. The stopping condition is $6 \mid c(S_t)$.

\medskip
\noindent For the first stopping level at $t = 3$, we have

\begin{center}
\begin{tabular}{ccc}
\toprule
$(k_1,k_2,k_3)$ & $c(k_1,k_2,k_3)$ & $c(k_1,k_2,k_3)\bmod 6$ \\
\midrule
$(3,0,0)$ and permutations & $1$ & $1$ \\
$(2,1,0)$ and permutations & $3$ & $3$ \\
$(1,1,1)$ & $6$ & $0$ \\
\bottomrule
\end{tabular}
\end{center}

The unique stopping state at $t = 3$ is $(1,1,1)$. The 6 paths are all permutations of $\{0,1,2\}$:

\begin{center}
\begin{tabular}{ccc}
\toprule
\textbf{Group $A_j$} & \textbf{Path} & \textbf{Sequence} \\
\midrule
$A_1$ & 1 & $0\ 1\ 2$ \\
$A_2$ & 2 & $0\ 2\ 1$ \\
$A_3$ & 3 & $1\ 0\ 2$ \\
$A_4$ & 4 & $1\ 2\ 0$ \\
$A_5$ & 5 & $2\ 0\ 1$ \\
$A_6$ & 6 & $2\ 1\ 0$ \\
\bottomrule
\end{tabular}
\end{center}

Each singleton group has probability $\pi_0\pi_1\pi_2$; the walk stops optimally in just 3 draws whenever all three symbols appear in the first round.

\medskip
For the continuation at $t = 4$, if the first 3 draws are, say, $0,0,1$ then $S_3 = (2,1,0)$ and $c(2,1,0) = 3 \not\equiv 0 \pmod 6$. At $t = 4$:

\begin{center}
\begin{tabular}{ccc}
\toprule
$(k_1,k_2,k_3)$ (up to permutation) & $c(k_1,k_2,k_3)$ & $c(k_1,k_2,k_3) \bmod 6$ \\
\midrule
$(4,0,0)$ & 1 & 1 \\
$(3,1,0)$ & 4 & 4 \\
$(2,2,0)$ & 6 & 0 \\
$(2,1,1)$ & 12 & 0 \\
\bottomrule
\end{tabular}
\end{center}

Both $(2,2,0)$ and $(2,1,1)$, together with their permutations, therefore lie in $H_6$ and stop the walk at $t = 4$.

For the state $(2,1,1)$, the total number of paths is $c = 12$. However, six of these twelve paths pass through $(1,1,1) \in H_6$ at $t = 3$, specifically those whose fourth symbol is
the repeated~$0$, and would have terminated at that point. The observable paths are the $m_E((2,1,1)) = 6$ avoiding ones, listed in lexicographic order as $0012$, $0021$, $0102$, $0201$, $1002$, and $2001$; these form six singleton
groups. The observed path $0,0,1,2$ has avoiding rank~$0$, which yields
group~$A_1$.

\medskip
Composite modulus $m$'s irregularity makes the group sizes $m_E(x)/6$ vary, and more notably, some states of $H_6$ are unreachable. Every path to $(2,2,1)$ encounters $H_6$ at an earlier stage, so $m_E((2,2,1)) = 0$ and the walk never terminates at this state. Furthermore, no stopping is possible at $t = 5$, since all three $H_6$-states at that level are unreachable; the same situation occurs for $(2,2,2)$ at $t = 6$. This irregularity is the primary obstacle to obtaining a closed-form expression for $E(T)$ when $m$ is composite (see Section~\ref{sec:complexity}).

\begin{center}
\begin{tabular}{ccccc}
\toprule
$t$ & Representative state & $c(S_t)$ & $m_E(S_t)$ & Group size \\
\midrule
3 & $(1,1,1)$ & 6 & 6 & 1 \\
4 & $(2,1,1)$ & 12 & 6 & 1 \\
4 & $(2,2,0)$ & 6 & 6 & 1 \\
5 & $(2,2,1)$ & 30 & 0 & unreachable \\
6 & $(3,2,1)$ & 60 & 12 & 2 \\
6 & $(4,1,1)$ & 30 & 18 & 3 \\
6 & $(2,2,2)$ & 90 & 0 & unreachable \\
\bottomrule
\end{tabular}
\end{center}

\section{Complexity and Entropy}
\label{sec:complexity}

The exact expected stopping time of Algorithm~\ref{alg:fairsample} for prime modulus $m = p$ is now derived. The following theorem is Corollaire~4 in \S3 of~\cite{bernardletac}.

\begin{theorem}
\label{thm:ET}
Let $m = p$ be a prime and $\bfpi$ any non-degenerate distribution on $I$. Define $U_n(\bfpi) = \sum_{i \in I} \pi_i^{p^n}$.
\[
 E(T) = p \prod_{n=1}^{\infty} \frac{1 - (U_n(\bfpi))^p}{1 - U_n(\bfpi)}\,.
\]
\end{theorem}

\begin{proof}
See \S3 of~\cite{bernardletac}. The proof follows from the generating function identity (\S3 of~\cite{bernardletac}) using the tail-sum representation $E(T) = \sum_{t \geq 0} P(T > t)$.
\end{proof}

\begin{remarknum}\label{rem:simple-approx}
A na\"ive renewal argument at the first scale $t = p$ gives
\begin{equation}\label{eq:simple-approx}
 E(T) \;\approx\; \frac{p}{1 - U_1(\bfpi)}
 \;=\; \frac{p}{1 - \displaystyle\sum_{i \in I} \pi_i^p}\,.
\end{equation}
This is a renewal estimate, not a truncation of the product to its $n = 1$ factor. Retaining only that factor would give $p(1 - U_1^p)/(1 - U_1)$, a different quantity. The discrepancy arises because the shifted walk from $p\,\eps^i$ does not exhibit the same carry structure as a new walk. When the position-$1$ digit in base~$p$ is non-zero, carries from digit~$0$ may propagate to digit~$1$ and higher, resulting in additional stopping events not present in the original walk. (For $p = 2$ with a fair binary source, the simple formula yields~$4$ while the exact product yields~$3.40$.)

The accuracy of the approximation increases rapidly as $p$ grows; the two correction factors relating~\eqref{eq:simple-approx} to the exact product are quantified in the proof of Theorem~\ref{thm:ET-renyi}. Empirically, the simple formula is accurate to within $1\%$ for all non-degenerate $\bfpi$ once $p \geq 5$, while at $p = 3$ the error can exceed $1\%$; like the unimodality observed in Section~\ref{sec:unimodality}, this accuracy claim is numerical and we do not prove it. Theorem~\ref{thm:ET-renyi} shows that the simple formula always overestimates $E(T)$.
\end{remarknum}

For the uniform source the product specializes as follows.

\begin{proposition}
\label{prop:ET-uniform}
Under the uniform distribution $\pi_i = 1/|I|$ for all $i$:
\[
 E(T) = p \prod_{n=1}^{\infty}
 \frac{1 - |I|^{p(1-p^n)}}{1 - |I|^{1-p^n}}\,.
\]
For sufficiently large $|I|$ or $p$, the expression is well approximated by $p/(1 - |I|^{1-p})$.
\end{proposition}

The following proposition summarizes the principal monotonicity properties.

\begin{proposition}
\label{prop:ET-monotone}
Let $p$ be a prime and $\bfpi$ a non-degenerate distribution on a finite or countable alphabet~$I$. The quantity $E(T)$ of Theorem~\ref{thm:ET}, which depends on both $\bfpi$ and~$p$, satisfies:
\begin{enumerate}[label=\emph{(\alph*)}]
\item $E(T) > p$ strictly, with equality approached for uniform sources as $|I| \to \infty$ and, for fixed $\bfpi$, as $p \to \infty$.
\item $E(T)$ is Schur-convex in $\bfpi$; it therefore decreases under majorization towards the uniform distribution, and along the binary family $\bfpi = (a, 1-a)$ it is strictly decreasing in $\Hsh(\bfpi)$.
\item $E(T) \to \infty$ as $\pi_i \to 1$ for any fixed $i$
(near-degenerate sources).
\item For fixed $\bfpi$, $E(T)/p \to 1$ as $p \to \infty$.
\end{enumerate}
\end{proposition}

\begin{proof}
Write $g(u) = (1 - u^p)/(1-u) = 1 + u + \cdots + u^{p-1}$, which is strictly increasing on $(0,1)$ and exceeds~$1$ there, so that
\begin{equation}\label{eq:logET}
 \log E(T) = \log p + \sum_{n \geq 1} \log g\bigl(U_n(\bfpi)\bigr).
\end{equation}
Since $U_n(\bfpi) > 0$ every factor is strictly greater than $1$, giving the strict inequality in~(a); the stated limits hold because $U_n \to 0$ both for uniform $\bfpi$ as $|I| \to \infty$ and for fixed $\bfpi$ as $p \to \infty$.

For part~(b), each $U_n(\bfpi) = \sum_i \pi_i^{p^n}$ is a coordinatewise sum of the convex function $t \mapsto t^{p^n}$ and is therefore Schur-convex; composing with the increasing map $\log g$ preserves Schur-convexity, and a sum of Schur-convex functions is Schur-convex, so~\eqref{eq:logET} makes $\log E(T)$ Schur-convex in $\bfpi$; since the exponential is increasing, $E(T)$ is Schur-convex as well. On the binary family $\bfpi = (a,1-a)$ with $a \geq \tfrac12$, majorization is a total order along which $\Hsh$ is strictly decreasing, so $E(T)$ is there strictly decreasing in $\Hsh(\bfpi)$. Note that no such statement can hold for general alphabets in terms of $\Hsh$ alone, since $E(T)$ depends on all the R\'enyi entropies $H_{p^n}$. The source $\bfpi = (0.45,0.45,0.1)$ has $\Hsh = 0.94892$ and $E(T) = 3.07269$ at $p = 2$, whereas $\bfpi' = (0.6,0.2,0.2)$ has the larger $\Hsh = 0.95027$ and the larger $E(T) = 3.31821$.

Part~(c) holds because $U_n \to 1$ as $\pi_i \to 1$ for every fixed~$n$, while $g$ increases to~$g(1^-) = p$; hence for every fixed~$N$, $\log E(T) \geq \log p + \sum_{n=1}^{N} \log g(U_n) \to (N+1)\log p$, and $E(T) \to \infty$, although no single factor diverges. For part~(d), put $q = \max_i \pi_i < 1$; then $U_1 = \sum_i \pi_i\,\pi_i^{\,p-1} \leq q^{\,p-1}$, a bound valid also for countable~$I$, and $U_n = \sum_i (\pi_i^{\,p^{n-1}})^{p} \leq U_{n-1}^{\,p}$, so $U_n \leq U_1^{\,p^{n-1}} \leq U_1^{\,n}$ and $\sum_{n \geq 1} U_n \leq U_1/(1 - U_1)$. Since $g(u) \leq 1 + (p-1)u$, once $p$ is large enough that $U_1 \leq \tfrac{1}{2}$ we get $0 \leq \log\bigl(E(T)/p\bigr) = \sum_{n \geq 1} \log g(U_n) \leq (p-1)\sum_{n \geq 1} U_n \leq 2(p-1)\,q^{\,p-1} \to 0$.
\end{proof}

The computational cost of the sampler is summarized next.

\begin{proposition}
\label{prop:per-step}
Along a run of Algorithm~\ref{alg:fairsample} that stops at time~$T$, the stopping test costs $O(|I| \cdot \omega(m) \cdot \log t)$ operations at step~$t$, where $\omega(m)$ denotes the number of distinct prime factors of~$m$, so the total cost of the walk is
\[
 O\!\left(|I| \cdot \omega(m) \cdot T\log T\right)
\]
operations, and $O\!\left(|I| \cdot \omega(m) \cdot E(T\log T)\right)$ in expectation, where the expectation is finite because $T$ has a geometric tail. The path-assignment step, which ranks the observed path among the $m_E(S_T)$ paths avoiding $H_m$ using the backward dynamic program of Algorithm~\ref{alg:fairsample}, incurs an additional $O\bigl(|I| \cdot \prod_{i}(x_i + 1)\bigr)$ arithmetic operations and requires storing integers at termination. These integers have $O(T \log |I|)$ bits, since every count $A(v)$ is at most $|I|^{T}$, where $x = S_T$. This cost is polynomial in $T$ for fixed $|I|$ but exponential in the alphabet size. For every prime modulus $p$, the transfer theorem of Section~\ref{sec:transfer} reduces the assignment cost to $O(|I|\,T \log_p T)$ automaton transitions (Corollary~\ref{cor:P1-binary}), which for $m = 2$ on a binary alphabet is $O(T \log T)$ bit operations; composite moduli remain Open Problem~P1 (Section~\ref{sec:open}).
\end{proposition}

The geometric tail is seen as follows. Let $i$ and $j$ be two symbols of positive probability and $\rho = \min(\pi_i, \pi_j)$. From any state, at most $2m$ further draws of the symbols $i$ and $j$ bring the walk to a state with $x_i \equiv -1$ and $x_j \equiv 1 \pmod m$, at which every prime power $p^a$ dividing $m$ produces $a$ carries in base~$p$ and Theorem~\ref{thm:kummer} gives $m \mid c(x)$; hence the walk enters $H_m$ within $2m$ draws with probability at least $\rho^{2m}$, and $P(T > 2mk) \leq (1 - \rho^{2m})^k$ for every $k \geq 0$. Since $\omega(m) = O(\log m / \log \log m)$, the cost of step~$t$ is $O(|I| \log t \log m)$. For prime $m = p$ one has $\omega(p) = 1$, and since $E(T)$ is of order~$p$ by Proposition~\ref{prop:ET-monotone}(d), the cost is $O(|I| \log p)$ per draw on average.

\begin{remarknum}[No closed form for composite $m$]
\label{prop:composite}
For composite $m = p_1^{a_1} \cdots p_r^{a_r}$ no analogue of the closed form of Theorem~\ref{thm:ET} is known, and its derivation does not extend. The generating-function argument behind Theorem~\ref{thm:ET} rests on the product structure of the avoiding path counts over the base-$p$ digit positions (\S3 of~\cite{bernardletac}), which supposes a single prime base. The stopping set $H_m = \bigcap_{j} H_{p_j^{a_j}}$ couples the digit and carry structures of the distinct bases $p_j$, and these are not simultaneously controllable; even within a single prime base there is no exact renewal regeneration (Remark~\ref{rem:simple-approx}). The irregularities of Example~2 in Section~\ref{sec:examples}, where $H_6$ contains unreachable states and the group sizes $m_E(x)/m$ vary, are symptoms of this coupling.
\end{remarknum}

One might hope to recover a bound of the form $E(T) = O(m)$ by observing that by time $t = m$ most draw patterns have a multinomial coefficient divisible by some prime factor of~$m$. This is not sufficient. Membership in $H_m$ requires the valuation conditions $v_{p_j}(c(x)) \geq a_j$ to hold \emph{simultaneously} for all~$j$. At $t = 6$, for instance, the state $(4,2)$ has $c = 15$, which is divisible by~$3$ but not by~$2$ and so does not lie in $H_6$; the axis states have $c = 1$ at every~$t$. Moreover no bound $O(m)$ can hold uniformly in $\bfpi$, since $E(T) \to \infty$ as $\bfpi$ degenerates for each fixed~$m$. For fixed non-degenerate $\bfpi$, we expect $E(T) = O_{\bfpi}(m)$, and the benchmark of Section~\ref{sec:benchmark} is consistent with this, but we do not prove it; see Open Problem~P2 (Section~\ref{sec:open}).

Remark~\ref{prop:composite} motivates two questions addressed in Section~\ref{sec:open}. The first asks whether a formula for $E(T)$ in terms of the prime-power R\'{e}nyi entropies nonetheless exists (Open Problem~P2). The second asks whether, given that prime moduli admit closed-form analysis, rounding $m$ up to the nearest prime is preferable in practice (Open Problem~P3).

We now establish the precise relationship between the expected stopping time
and the information-theoretic properties of the source.

\begin{definition}
The R\'enyi entropy of order $\alpha > 0$, $\alpha \neq 1$ of a distribution $\bfpi$ on $I$ is
\[
 H_\alpha(\bfpi) = \frac{1}{1 - \alpha} \log \sum_{i \in I} \pi_i^\alpha,
\]
with $H_1(\bfpi) = \lim_{\alpha \to 1} H_\alpha(\bfpi) = -\sum_i \pi_i \log \pi_i$ (the Shannon entropy).
\end{definition}

The entropy reading of Theorem~\ref{thm:ET}, together with its first-order approximation, can now be stated.

\begin{theorem}
\label{thm:ET-renyi}
For prime $m = p$, let $U_n(\bfpi) = \sum_{i \in I} \pi_i^{p^n}$ and note that $U_1(\bfpi) = e^{(1-p)\,H_p(\bfpi)}$ where $H_p$ is the R\'enyi entropy
of order~$p$. Then:
\begin{enumerate}[label=\emph{(\alph*)}]
 \item \textbf{Exact:}\quad
 $\displaystyle E(T) = p \prod_{n=1}^{\infty}
 \frac{1 - (U_n(\bfpi))^p}{1 - U_n(\bfpi)}$.
 \item \textbf{First-order approximation:}\quad
 $\displaystyle E(T) \approx \frac{p}{1 - e^{(1-p)\,H_p(\bfpi)}}$.
\end{enumerate}
The approximation~\emph{(b)} is a renewal estimate at the first scale $t = p$ rather than a truncation of the product, and it always overestimates, in the sense that $E(T) < p/(1 - U_1(\bfpi))$ strictly for every non-degenerate~$\bfpi$; in particular the estimate is never exact.
\end{theorem}

\begin{proof}
Part (a) is Theorem~\ref{thm:ET}. For part (b), split the product at $n = 1$ and rearrange to obtain the exact identity
\[
 \frac{p}{1 - U_1} \;=\; E(T)\cdot\frac{1}{1 - U_1^{\,p}}\cdot\prod_{n \geq 2}\frac{1 - U_n}{1 - U_n^{\,p}}\,,
\]
so the estimate~\eqref{eq:simple-approx} is $E(T)$ multiplied by two correction factors, each tending to~$1$. For the first, $U_1 = \sum_i \pi_i^{p} = \sum_i \pi_i\,\pi_i^{p-1} \leq (\max_i \pi_i)^{p-1}$, whence $1 \leq 1/(1 - U_1^{\,p}) \leq \bigl(1 - (\max_i \pi_i)^{p(p-1)}\bigr)^{-1}$, which tends to~$1$ as $p$ grows. For the second, $U_n = \sum_i (\pi_i^{p^{n-1}})^{p} \leq \bigl(\sum_i \pi_i^{p^{n-1}}\bigr)^{p} = U_{n-1}^{\,p}$, so $U_n \leq U_1^{\,p^{n-1}}$ and the tail converges to~$1$ doubly exponentially. For the strict inequality, note that for every $N \geq 1$
\[
 (1 - U_1)\prod_{n=1}^{N}\frac{1 - U_n^{\,p}}{1 - U_n}
 \;=\; \bigl(1 - U_N^{\,p}\bigr)\prod_{n=1}^{N-1}\frac{1 - U_n^{\,p}}{1 - U_{n+1}}\,;
\]
letting $N \to \infty$ and using $U_N \to 0$, established just above, gives
\[
 \frac{E(T)}{p/(1 - U_1)} \;=\; \prod_{n \geq 1}\frac{1 - U_n^{\,p}}{1 - U_{n+1}}\,.
\]
Since $\bfpi$ is non-degenerate, at least two of the $\pi_i$ are positive, so $U_{n+1} = \sum_i (\pi_i^{\,p^n})^{p} < \bigl(\sum_i \pi_i^{\,p^n}\bigr)^{p} = U_n^{\,p}$ strictly for every $n \geq 1$. Hence every factor in the last product lies in $(0,1)$, and the product, which is positive because it equals $E(T)(1 - U_1)/p > 0$, is strictly less than~$1$.
\end{proof}

A key observation is that the complexity of the Bernard--Letac algorithm is governed by the R\'enyi entropies of the source at orders $p, p^2, p^3, \ldots$, with the principal contribution from the R\'enyi entropy associated with the modulus. The theorems above are stated without assuming $I$ is finite, and Bernard and Letac themselves work throughout with $I$ finite or countable, so the product formula and its R\'enyi reading apply verbatim on countable alphabets. What is not settled at the boundary, for sources of infinite Shannon entropy, is a replacement for the lower bound of Theorem~\ref{thm:lower-bound}; this is Open Problem~P6 (Section~\ref{sec:open}).

The limit of $E(T)$ as $p \to 1^+$ is determined next. The following lemma collects the elementary bounds on which the argument rests.

\begin{lemma}
\label{lem:g-bounds}
Let $\eps > 0$, $p = 1 + \eps$, and for $u \in (0,1)$ put
\[
  g(u) = \frac{1 - u^{p}}{1 - u}, \qquad
  f(u) = \frac{-\,u\log u}{1 - u}, \qquad
  G(u) = \frac{u\log^2\! u}{2(1 - u)} + \frac{f(u)^2}{2}.
\]
Then for every $u \in (0,1)$,
\begin{equation}
\label{eq:g-sandwich}
  0 \;\leq\; \eps f(u) - \log g(u) \;\leq\; \eps^2 G(u) \;\leq\; \eps^2 .
\end{equation}
Moreover $f$ is strictly increasing on $(0,1)$ with $0 < f(u) < 1$, $f(u) \to 0$ as $u \to 0^+$, and $f(u) \to 1$ as $u \to 1^-$.
\end{lemma}

\begin{proof}
Since $u^{p} = u\,u^{\eps}$, one has $g(u) = 1 + x$ with $x = u(1 - u^{\eps})/(1-u) > 0$. Put $t = \eps\log(1/u) > 0$, so that $u^{\eps} = e^{-t}$. The inequalities $t - t^2/2 \leq 1 - e^{-t} \leq t$, valid for $t \geq 0$, give
\[
  \eps f(u) - \eps^2\,\frac{u\log^2\! u}{2(1-u)} \;\leq\; x \;\leq\; \eps f(u).
\]
The inequalities $x - x^2/2 \leq \log(1 + x) \leq x$, valid for $x \geq 0$, then give $\log g(u) \leq x \leq \eps f(u)$, which is the first inequality in~\eqref{eq:g-sandwich}, and
\[
  \log g(u) \;\geq\; x - \frac{x^2}{2}
  \;\geq\; \eps f(u) - \eps^2\,\frac{u\log^2\! u}{2(1-u)} - \frac{\eps^2 f(u)^2}{2}
  \;=\; \eps f(u) - \eps^2 G(u),
\]
which is the second. For the third, the inequality $\log y \leq (y - 1/y)/2$, valid for $y \geq 1$, applied at $y = u^{-1/2}$ gives $\log(1/u) \leq (1 - u)/\sqrt{u}$, hence $u\log^2\! u/(1-u) \leq 1 - u$, and together with $f(u) < 1$, proved below, this yields $G(u) \leq (1-u)/2 + 1/2 \leq 1$.

The derivative of $f$ is $f'(u) = \bigl(\log(1/u) - (1 - u)\bigr)/(1-u)^2$, which is positive on $(0,1)$ because $\log y > 1 - 1/y$ for $y > 1$. The bound $f(u) < 1$ is equivalent to $\log(1/u) < 1/u - 1$, which is $\log y < y - 1$ for $y = 1/u > 1$. Finally, $u\log u \to 0$ as $u \to 0^+$ gives $f(u) \to 0$, and $\log u \sim u - 1$ as $u \to 1^-$ gives $f(u) \to 1$.
\end{proof}

\begin{theorem}
\label{thm:shannon-limit}
Treating $p > 1$ as a continuous parameter, the expected stopping time
$E(T) = p\prod_{n\geq 1}(1 - U_n^p)/(1 - U_n)$, with
$U_n = \sum_{i\in I}\pi_i^{p^n}$, has a finite limit as $p \to 1^+$:
\[
  \lim_{p \to 1^+} E(T) \;=\; C(\bfpi)
  \;:=\; \exp\!\left(\int_0^\infty
      \frac{-\,\varphi(s)\,\log\varphi(s)}{1 - \varphi(s)}\,ds\right),
  \qquad
  \varphi(s) = \sum_{i\in I}\pi_i^{\,e^s}.
\]
The constant $C(\bfpi)$ is strictly greater than~$1$, depends on the full
distribution $\bfpi$ rather than on its Shannon entropy alone, and is
not in general equal to $1/\Hsh(\bfpi)$. For the symmetric binary
source $\bfpi = (\tfrac12,\tfrac12)$ one has $C(\bfpi) = 2.8457\ldots$,
whereas $1/\Hsh(\bfpi) = 1.4427\ldots$.
\end{theorem}

Consequently, the Bernard--Letac sampler does not achieve the Shannon optimum as $p \to 1^+$. The lower bound in Theorem~\ref{thm:lower-bound} scales as $(\log p)/\Hsh(\bfpi) \to 0$, while $E(T)$ remains bounded below by the positive constant $C(\bfpi)$. Thus, the efficiency approaches zero (Corollary~\ref{cor:efficiency}). The naive reduction of the product to its single $n = 1$ factor is invalid because the tail $\prod_{n\geq 2}$ consists of infinitely many factors, each approaching $1$, yet their product converges to a nontrivial constant. This phenomenon is the same carry-coupling effect that differentiates the exact product in Theorem~\ref{thm:ET} from the single-scale approximation in Remark~\ref{rem:simple-approx}.

\begin{proof}
Write $p = 1 + \eps$ with $\eps \to 0^+$, and set $a_i = -\log\pi_i > 0$, so that $U_n = \sum_i e^{-a_i p^n} = \varphi(n\log p)$ with $\varphi(s) = \sum_i \pi_i^{e^s}$. Taking logarithms,
\[
  \log E(T) = \log p + \sum_{n\geq 1}\log g(U_n),
\]
with $g$ as in Lemma~\ref{lem:g-bounds}. Let $f$ and $G$ be as in that lemma, put $F(s) = f(\varphi(s)) = -\varphi(s)\log\varphi(s)/(1 - \varphi(s))$ for $s > 0$, and put $s_n = n\log p$ and $\Delta = \log p$, so that $f(U_n) = F(s_n)$ and the points $s_n$ form a mesh of width $\Delta$ on $(0,\infty)$.

The function $\varphi$ is continuous and strictly decreasing on $[0,\infty)$, since each $\pi_i^{e^s}$ is continuous and decreasing in $s$ and is dominated by $\pi_i$, with $\varphi(0) = 1$ and $0 < \varphi(s) < 1$ for $s > 0$. By Lemma~\ref{lem:g-bounds}, $F$ is therefore strictly decreasing on $(0,\infty)$ with $0 < F < 1$, $F(s) \to 1$ as $s \to 0^+$, and $F(s) \to 0$ as $s \to \infty$. Moreover $F$ is integrable on $(0,\infty)$. Indeed, for $u \leq 1/2$ one has $f(u) \leq 2u\log(1/u) \leq 4\sqrt{u}$, the last step because $\log(1/u) = 2\log(u^{-1/2}) \leq 2(u^{-1/2} - 1)$, while $\varphi(s) \leq M^{e^s - 1}$ with $M = \max_i \pi_i < 1$, valid also for countable~$I$ since $\pi_i^{e^s} \leq \pi_i M^{e^s - 1}$. Hence $F(s) \leq 4M^{(e^s - 1)/2}$ once $\varphi(s) \leq 1/2$, which decays doubly exponentially, and $F \leq 1$ before that. The integral $\int_0^\infty F(s)\,ds$ is thus finite, and it is strictly positive because $F > 0$.

Since $F$ is decreasing and nonnegative, $F(s_n) \leq \Delta^{-1}\int_{s_{n-1}}^{s_n} F(s)\,ds$ and $F(s_n) \geq \Delta^{-1}\int_{s_n}^{s_{n+1}} F(s)\,ds$ for every $n \geq 1$, where $s_0 = 0$. Summing these inequalities gives
\begin{equation}
\label{eq:riemann-sandwich}
  \int_{\Delta}^\infty F(s)\,ds \;\leq\; \Delta\sum_{n\geq 1} F(s_n) \;\leq\; \int_0^\infty F(s)\,ds,
  \qquad
  \Delta\sum_{n > N} F(s_n) \;\leq\; \int_{N\Delta}^\infty F(s)\,ds
\end{equation}
for every $N \geq 0$. In particular the series $\sum_n F(s_n)$ converges, and since $F \leq 1$ the first part of~\eqref{eq:riemann-sandwich} yields
\[
  0 \;\leq\; \int_0^\infty F(s)\,ds - \Delta\sum_{n\geq 1} F(s_n) \;\leq\; \int_0^{\Delta} F(s)\,ds \;\leq\; \Delta .
\]

Now apply~\eqref{eq:g-sandwich} at $u = U_n \in (0,1)$, which gives $0 \leq \eps F(s_n) - \log g(U_n) \leq \eps^2$ for every $n \geq 1$. In particular $0 < \log g(U_n) \leq \eps F(s_n)$, so the series $\sum_n \log g(U_n)$ converges. Let $N = \floor{\eps^{-3/2}}$. Summing the two-sided bound over $n \leq N$ and the one-sided bound $0 \leq \eps F(s_n) - \log g(U_n) \leq \eps F(s_n)$ over $n > N$ gives
\[
  0 \;\leq\; \eps\sum_{n\geq 1} F(s_n) - \sum_{n\geq 1}\log g(U_n)
  \;\leq\; \eps^2 N + \eps\sum_{n > N} F(s_n)
  \;\leq\; \eps^{1/2} + \frac{\eps}{\Delta}\int_{N\Delta}^\infty F(s)\,ds ,
\]
by the second part of~\eqref{eq:riemann-sandwich}. As $\eps \to 0^+$, one has $\eps/\Delta = \eps/\log(1+\eps) \to 1$ and $N\Delta \geq (\eps^{-3/2} - 1)\log(1 + \eps) \to \infty$, so the right side tends to~$0$ because $F$ is integrable, while
\[
  \eps\sum_{n\geq 1} F(s_n) = \frac{\eps}{\Delta}\Bigl(\int_0^\infty F(s)\,ds + O(\Delta)\Bigr) \;\longrightarrow\; \int_0^\infty F(s)\,ds .
\]
Consequently $\sum_{n\geq 1}\log g(U_n) \to \int_0^\infty F(s)\,ds$, and since $\log p \to 0$,
\[
  \log E(T) \;\longrightarrow\; \int_0^\infty F(s)\,ds \;>\; 0 .
\]
Exponentiating yields $\lim_{p\to 1^+} E(T) = C(\bfpi)$ and $C(\bfpi) > 1$. For the symmetric binary source, $\varphi(s) = 2^{\,1 - e^s}$, and numerical quadrature of the integral gives $C(\bfpi) = 2.84575\ldots$. The constant is not a function of the Shannon entropy alone. For instance, the uniform source on four symbols and the source $(\tfrac12, \tfrac18, \tfrac18, \tfrac18, \tfrac18)$ both have Shannon entropy $2\log 2$, while numerical quadrature gives $C(\bfpi) = 1.98434\ldots$ for the former and $C(\bfpi) = 2.29289\ldots$ for the latter.
\end{proof}

The following theorem gives the information-theoretic lower bound.

\begin{theorem}
\label{thm:lower-bound}
For any algorithm in $\mathfrak{A}_m$,
\[
 E(T) \;\geq\; \frac{\log m}{\Hsh(\bfpi)}\,.
\]
\end{theorem}

\begin{proof}
If $E(T) = \infty$ there is nothing to prove, so assume $E(T) < \infty$. Since the source is i.i.d.\ and $T$ is a stopping time with finite expectation, Wald's identity applied to the variables $-\log \pi_{a_t}$ gives the entropy identity
\[
 \Hsh(a_1,\ldots,a_T) \;=\; E(T)\,\Hsh(\bfpi).
\]
The output $J$ is $\mathcal{F}_T$-measurable, so $\Hsh(J) \leq \Hsh(a_1,\ldots,a_T)$, while $J$ is uniform on $m$ values and hence $\Hsh(J) = \log m$. Combining the two gives $\log m \leq E(T)\,\Hsh(\bfpi)$.
\end{proof}

The following proposition shows that the lower bound is never attained.

\begin{proposition}
\label{prop:lower-bound-strict}
For every $m \geq 2$, every non-degenerate distribution $\bfpi$, and every algorithm in $\mathfrak{A}_m$, the inequality of Theorem~\ref{thm:lower-bound} is strict.
\end{proposition}

\begin{proof}
Write $P_{\bfpi}$ for the law of the draw sequence under an admissible $\bfpi$, and suppose that equality holds for some algorithm in $\mathfrak{A}_m$ at some non-degenerate $\bfpi^0$. If $\Hsh(\bfpi^0) = \infty$ the right-hand side vanishes while $E(T) \geq 1$, and if $\Hsh(\bfpi^0) < \infty$ while $E(T) = \infty$ the left-hand side alone is infinite, so in either case the inequality is strict; assume therefore that both quantities are finite. Equality forces equality in the step $\Hsh(J) \leq \Hsh(a_1, \ldots, a_T)$ of the proof of Theorem~\ref{thm:lower-bound}, and since $J$ is a function of the stopped trajectory, this reads $\Hsh(a_1, \ldots, a_T \mid J) = 0$ under $P_{\bfpi^0}$. Hence for each $j$ there is a single word $w_j$, of length $t_j \geq 1$ and letter-count vector $x_j \in \Mon$ with $\abs{x_j} = t_j$, such that $A_j$ and $\{(a_1, \ldots, a_T) = w_j\}$ differ by a $P_{\bfpi^0}$-null set; the $w_j$ are pairwise distinct, since otherwise two disjoint events of probability $1/m$ would coincide up to null sets.

Write $[w] = \{a_1 = w^{(1)}, \ldots, a_t = w^{(t)}\}$ for the cylinder of draw sequences beginning with the word $w = w^{(1)} \cdots w^{(t)}$ of length~$t$. The event $\{T = t_j\}$ lies in $\mathcal{F}_{t_j}$, so its intersection with $[w_j]$ is either empty or all of $[w_j]$. This intersection is the event $\{(a_1, \ldots, a_T) = w_j\}$, of $P_{\bfpi^0}$-probability $1/m > 0$, hence it equals $[w_j]$ as a set, and in particular $(\bfpi^0)^{x_j} = P_{\bfpi^0}([w_j]) = 1/m$, while $P_{\bfpi}([w_j]) = \bfpi^{x_j}$ for every admissible~$\bfpi$. Likewise $A_j \cap [w_j] = A_j \cap \{T = t_j\} \cap [w_j]$ lies in $\mathcal{F}_{t_j}$, is empty or all of $[w_j]$, and has $P_{\bfpi^0}$-probability $1/m$, whence $[w_j] \subseteq A_j$ as sets. Consequently, for every admissible $\bfpi$,
\[
 \bfpi^{x_j} \;=\; P_{\bfpi}([w_j]) \;\leq\; P_{\bfpi}(A_j) \;=\; \frac{1}{m}\,.
\]

If $w_j$ used only one letter $i$, the left-hand side would approach $1$ as $\pi_i$ increases to $1$ within the admissible distributions, so every $w_j$ uses at least two distinct letters. The maximum of the monomial $\bfpi \mapsto \bfpi^{x_j}$ over the whole simplex is then attained at the point $x_j/t_j$, which has all coordinates strictly below $1$ and is therefore admissible, and the maximum equals $e^{-t_j \Hsh(x_j/t_j)}$, uniquely attained by the Gibbs inequality~\cite[Chapter~XI]{gibbs}. The display gives $1/m = (\bfpi^0)^{x_j} \leq e^{-t_j \Hsh(x_j/t_j)} \leq 1/m$, so by uniqueness $\bfpi^0 = x_j/t_j$, whence $1/m = e^{-t_j \Hsh(\bfpi^0)}$, that is, $t_j = \log m/\Hsh(\bfpi^0)$. The $m$ distinct words $w_j$ therefore share one length $t$ and one letter-count vector $x = t\,\bfpi^0$, which forces $m \leq c(x)$. On the other hand, summing $P_{\bfpi^0}(S_t = y) = c(y)\,(\bfpi^0)^{y}$ over all $y$ with $\abs{y} = t$ gives $1 > c(x)\,(\bfpi^0)^{x} = c(x)/m$, the inequality being strict because the type $t\,\eps^{i}$, for any letter $i$ in the support of $\bfpi^0$, differs from $x$ and carries positive probability. Hence $c(x) < m$, a contradiction.
\end{proof}

In particular, the efficiency $\eta(p)$ of Corollary~\ref{cor:efficiency} satisfies $\eta(p) < 1$ for every prime $p$ and every non-degenerate $\bfpi$.

\begin{corollary}
\label{cor:efficiency}
The efficiency of the Bernard--Letac algorithm,
\[
  \eta(p) = \frac{\log p}{E(T)\cdot\Hsh(\bfpi)},
\]
regarded as a function of the continuous parameter $p > 1$ through the product formula, as in Theorem~\ref{thm:shannon-limit}, tends to~$0$ at both extremes of the scale. As $p$ approaches $1^+$, $E(T)$ converges to $C(\bfpi) > 0$ (Theorem~\ref{thm:shannon-limit}) while $\log p$ approaches zero. As $p$ approaches infinity, $E(T)$ is asymptotically equivalent to $p$ (Proposition~\ref{prop:ET-monotone}(d)), so $\eta(p)$ behaves as $\log p/(p\,\Hsh)$ and also approaches zero. Since $\eta$ is continuous and positive on $(1,\infty)$ with limits of zero at both ends, it attains a maximum at some interior point $p^\star = p^\star(\bfpi) \in (1,\infty)$. The sampler therefore achieves its highest efficiency at an intermediate scale. Under the first-order approximation (Theorem~\ref{thm:ET-renyi}(b)),
\[
  \eta(p) \;\approx\; \frac{(1 - e^{(1-p)H_p})\log p}{p\,\Hsh(\bfpi)}\,.
\]
\end{corollary}

For instance, with $\bfpi = (0.7, 0.3)$, the efficiency increases from near zero as $p$ approaches $1^+$ to a maximum value of $\eta \approx 0.44$ near $p^\star \approx 5$, and then decreases (Figure~\ref{fig:efficiency}). For the more skewed source $\bfpi = (0.9, 0.1)$, the maximum efficiency is $\eta \approx 0.46$ near $p^\star \approx 10$. Two distinct effects reduce efficiency at the extremes. For large $p$, $H_p(\bfpi) \leq \Hsh(\bfpi)$ because R\'enyi entropy is non-increasing in order, so the algorithm detects less randomness than is present in the source. As $p$ approaches $1^+$, the target contains vanishing information ($\log p \to 0$), while $E(T)$ remains bounded below by the positive constant $C(\bfpi)$, resulting in most draws being wasted.

\begin{proposition}
\label{prop:unified}
The following asymptotic regimes hold for the expected stopping time of Algorithm~\ref{alg:fairsample}:

\begin{center}
\begin{tabular}{lll}
\toprule
\textbf{Regime} & $E(T)$ & \textbf{Governing quantity} \\
\midrule
$p \to 1^+$ & $\to C(\bfpi) > 1$ & full law of $\bfpi$ (Thm~\ref{thm:shannon-limit}); $\neq 1/\Hsh$ \\
$p$ moderate & $= p\prod_{n\geq 1}\frac{1-U_n^p}{1-U_n}$ & R\'enyi-$p^n$ entropies \\
$p$ moderate (approx.) & $\approx p/(1 - e^{(1-p)H_p})$ & R\'enyi-$p$ entropy \\
$p \to \infty$ & $\sim p$ \ ($E(T)/p \to 1$) & floor; correction $\sim p\,e^{(1-p)H_p}$ \\
$\bfpi$ near-degenerate & $\to \infty$ & $\Hsh(\bfpi) \to 0$ \\
$\bfpi$ uniform & $\approx p/(1 - |I|^{1-p})$ & $H_p = \log|I|$ for all $p$ \\
\bottomrule
\end{tabular}
\end{center}
\end{proposition}

\section{A Transfer Theorem for the First-Passage Kernel}
\label{sec:transfer}

Proposition~\ref{prop:per-step} identifies group assignment as the computational bottleneck of Algorithm~\ref{alg:fairsample}. This section removes it for every prime modulus $m = p$ and every finite alphabet~$I$. It is demonstrated that the first-passage kernel underlying the ranking step can, modulo~$p$, be computed by an explicit automaton with $p + \abs{I} + 3$ states that processes the base-$p$ digits of its arguments. Consequently, each kernel value requires $O(\log_p T)$ automaton transitions, replacing the need for a dynamic program over the box $[\mathbf{0}, S_T]$, and the exact fair assignment incurs $O(\abs{I}\,T \log_p T)$ transitions in total (Corollary~\ref{cor:P1-binary}), which for $m = 2$ on the binary alphabet amounts to $O(T \log T)$ bit operations. That classical case, corresponding to the extraction of a single fair bit from a biased source, is worked out explicitly as a seven-state automaton and presented as self-contained pseudocode in Algorithm~\ref{alg:automaton}. Only one step of the proof does not constitute an exact path bijection, namely the free action of the cyclic group of order~$p$ on the orderings of a terminal mixed block, introducing the reduction modulo~$p$; for $p = 2$ this action is a fixed-point-free involution.

The term transfer is employed in the transfer-matrix sense, rather than the singularity-analysis sense from analytic combinatorics. The section operators defined below transfer a kernel evaluation at one base-$p$ scale to an evaluation at the next, enabling the global value $\kappa(v,z)$ to be constructed by composing local digit rules. Here, the digit expansion serves as the chain along which local transfer data are multiplied. Equivalently, Theorem~\ref{thm:transfer} establishes that the family of sections of $\kappa$ (its $p$-kernel in the terminology of automatic sequences~\cite{alloucheshallit}) is finite, which defines $p$-automaticity. The term kernel is reserved for $\kappa$ itself, while sections are used to distinguish between these concepts.

Throughout this section $p$ is a fixed prime and $I$ a finite alphabet with $n = \abs{I} \geq 2$ letters. We identify $\Mon$ with $\N^{I}$, equipped with the componentwise partial order, the unit steps $\eps^{i}$, and the coordinate sum $\abs{x} = \sum_{i \in I} x_i$, which we call the level of~$x$. A monotone path is a finite sequence of points in which each point is obtained from its predecessor by adding some~$\eps^{i}$. For $x \in \N^{I}$ and $\alpha \geq 0$ the digit column $h_\alpha(x) = (h_\alpha(x_i))_{i \in I}$ lies in $\Dig = \{0, \dots, p-1\}^{I}$; a column $d \in \Dig$ is called pure when at most one coordinate of $d$ is positive and mixed otherwise. For $e \in \N^{I}$ the multinomial coefficient $c(e) = \abs{e}!/\prod_i e_i!$ counts the orderings of a monotone path with increment~$e$. By the multinomial Lucas criterion following Theorem~\ref{thm:kummer},
\begin{align*}
 H_p &= \{x \in \N^{I} : c(x) \equiv 0 \pmod p\},\\
 \Cc &:= \N^{I} \setminus H_p = \{x : \abs{h_\alpha(x)} \leq p-1 \text{ for every } \alpha\},
\end{align*}
with $\mathbf{0} \in \Cc$, which defines the continuation region of the walk. For $p = 2$ and $I = \{0,1\}$ this reads $H_2 = \{z \neq \mathbf{0} : z_0 \mathbin{\&} z_1 \neq 0\}$ and $\Cc = \{u : u_0 \mathbin{\&} u_1 = 0\}$, where $\mathbin{\&}$ denotes bitwise \textsc{and}.

The first-passage kernel is defined next.

\begin{definition}
\label{def:kernel}
For $z \in H_p$ and $v \leq z$, let $\varphi(v,z)$ denote the number of monotone paths from $v$ to $z$ all of whose points except the terminal point lie in $\Cc$; such paths are called counted paths for the pair $(v,z)$. In particular $\varphi(z,z) = 1$ and $\varphi(v,z) = 0$ when $v \in H_p \setminus \{z\}$. The kernel is
\[
 \kappa(v,z) \;=\; \varphi(v,z) \bmod p \;\in\; \F_p,
\]
extended by $\kappa(v,z) = 0$ whenever $v \not\leq z$ or $z \notin H_p$. A kernel is any function $\N^{I} \times \N^{I} \to \F_p$.
\end{definition}

The relation $v \leq z$ is the componentwise order fixed in the notation of Section~\ref{sec:intro}. In contrast, the lexicographic order is used solely for ranking paths and for organizing the columns of Table~\ref{tab:transfer}. The extension $\kappa(v,z) = 0$ for $v \not\leq z$ simply indicates that no monotone path exists from $v$ to $z$.

The kernel takes values in $\F_p$ and is not normalized; thus, it does not constitute a stochastic kernel. However, the integer count $\varphi$ serves as the combinatorial foundation for such a kernel. From any starting point $v$, the walk reaches $H_p$ almost surely. Indeed, from every point of $\Cc$ at most $p$ suitably chosen draws produce a bottom digit column of sum~$p$, hence a point of $H_p$, and each block of $p$ draws realizes such a choice with probability at least $(\min_i \pi_i)^p > 0$. The probability that a walk starting at $v \in \Cc$ first enters $H_p$ at $z$ is given by $\varphi(v,z)\,\bfpi^{z-v}$, since each path from $v$ to $z$ that avoids $H_p$ has probability $\bfpi^{z-v}$. Therefore, $\sum_{z \in H_p} \varphi(v,z)\,\bfpi^{z-v} = 1$. The mapping $(v,z) \mapsto \varphi(v,z)\,\bfpi^{z-v}$ defines the stochastic first-passage (harmonic-measure) kernel for the stopped walk from $\Cc$ to $H_p$, with $\kappa$ preserving only the combinatorial structure.

The barrier being crossed at time $T$ is the combinatorial set $H_p$ rather than a scalar level, and $H_p$ is not absorbing for the unstopped walk. From $(1, p-1) \in H_p$ on the binary alphabet the step $\eps^1$ lands at $(1, p) \in \Cc$. Higher-order passages are therefore definable, but they play no role here, because the walk is stopped at $T$ and the convention $\varphi(v,z) = 0$ for $v \in H_p \setminus \{z\}$ forbids interior visits to $H_p$ in the count.

The kernel provides a unified framework for the two path counts discussed in this paper. For $x \in H_p$, $m_E(x) = \varphi(\mathbf{0}, x)$, so that part~(c) of Theorem~\ref{thm:BL-fundamental} involves the evaluation $\kappa(\mathbf{0}, x) = c(x) \bmod p$, while the ranking quantities in Algorithm~\ref{alg:fairsample} are given by $A(v) = \varphi(v, S)$. Two reductions ensure that the kernel modulo $p$ is sufficient to achieve exact fairness.

The first reduction is the following residue rule.

\begin{lemma}
\label{lem:residue}
For any modulus $m$, replacing the return value of Algorithm~\ref{alg:fairsample} by $J = (r \bmod m) + 1$ preserves exact fairness, in the sense that $P(J = j) = 1/m$ for every admissible $\bfpi$.
\end{lemma}

\begin{proof}
Conditional on $\{S_T = x\}$ the observed path is uniform over the $m_E(x)$ avoiding paths, by using the proof of Theorem~\ref{thm:alg2-correct}, whose ranks are the consecutive integers $0, \dots, m_E(x) - 1$. Since $m \mid m_E(x)$, each residue class modulo $m$ contains exactly
$m_E(x)/m$ of them.
\end{proof}

Second, the rank reduces to kernel evaluations along the observed path. Fix the linear order $\prec$ on $I$ used for ranking. The counted paths from $S_0$ to $S_T$ that precede the observed path $a = (a_1, \dots, a_T)$ lexicographically are those that agree with it before some time $t$ and take a letter $i \prec a_t$ at time~$t$. The branch points are therefore the points $S_{t-1} + \eps^{i}$ with $i \prec a_t$, and their counts $A(v) = \varphi(v, S_T)$ vanish for $v \in H_p \setminus \{S_T\}$, which is precisely the convention built into $\kappa$. It follows that
\begin{equation}
\label{eq:rank-kernel}
 r \;\equiv\; \sum_{t=1}^{T}\ \sum_{i \prec a_t}
  \kappa\bigl(S_{t-1} + \eps^{i},\; S_T\bigr) \pmod p .
\end{equation}
On the binary alphabet with $0 \prec 1$ the inner sum has a single term $\kappa(S_{t-1} + \eps^0, S_T)$ at the steps with $a_t = 1$ and no term otherwise.

\subsection*{Statement of the transfer theorem}

For digit columns $a, b \in \Dig$, define the section operator on kernels by $(\mathcal{S}_{a,b} K)(v,z) = K(pv + a,\, pz + b)$, so that iterated sections read the digit columns of the two arguments least significant first. A transfer system is a finite set of kernels containing $\kappa$ and closed under all $p^{2n}$ section operators.

The members of the transfer system are the following kernels. For $i \in I$ and $\lambda \in \F_p^{\times}$ set
\begin{align*}
 D^{i}(v,z) &= \mathbf{1}[v \in \Cc]\,\kappa(v + \eps^{i},\, z), \\
 \Delta_{H}(v,z) &= \delta_{v=z}\,\mathbf{1}[z \in H_p], \qquad
 \Delta_{\Cc}(v,z) = \delta_{v=z}\,\mathbf{1}[v \in \Cc],
\end{align*}
together with $\Delta = \delta_{v=z}$, the scalar multiples $\lambda\Delta_{\Cc}$, and the zero kernel~$0$. The kernels $D^{i}$ are called carry kernels, since they record an unfinished digit block in the letter~$i$, and the diagonal kernels are called certificates. Two linear identities relate them.

\begin{lemma}
\label{lem:kernel-identities}
$\Delta = \Delta_{H} + \Delta_{\Cc}$, and $\kappa = \sum_{i \in I} D^{i} + \Delta_{H}$.
\end{lemma}

\begin{proof}
The first identity holds because $H_p$ and $\Cc$ partition $\N^{I}$. For the second, evaluate at $(v,z)$. If $v = z$ both sides equal $\mathbf{1}[z \in H_p]$, since each $D^{i}(z,z)$ vanishes, the point $z + \eps^{i}$ not being below~$z$. If $v \neq z$ and $v \in H_p$, both sides vanish. If $v \neq z$ and $v \in \Cc$, the right side is $\sum_i \kappa(v + \eps^{i}, z)$, which is the first-step decomposition of the counted paths from~$v$; the convention $\kappa(x,z) = 0$ for a non-terminal $x \in H_p$ makes the decomposition exact. If $v \not\leq z$ both sides vanish.
\end{proof}

The transfer theorem can now be stated.

\begin{theorem}
\label{thm:transfer}
The family
\[
 \mathcal{F} \;=\; \{\kappa,\ 0\} \,\cup\, \{D^{i} : i \in I\} \,\cup\, \{\Delta_{H},\ \Delta\} \,\cup\, \{\lambda\Delta_{\Cc} : \lambda \in \F_p^{\times}\}
\]
is a transfer system for $\kappa$ of cardinality exactly $p + n + 3$. Its members are pairwise distinct, and each is a section of~$\kappa$. The sections of every member are given by Proposition~\ref{prop:table} and Lemma~\ref{lem:closure} below. Consequently $\kappa(v,z)$ is computed by the deterministic finite automaton with output whose states are the members of $\mathcal{F}$, whose initial state is $\kappa$, whose input is the sequence of digit columns $(h_\alpha(v), h_\alpha(z)) \in \Dig \times \Dig$ read for $\alpha = 0, 1, 2, \dots$ (least significant first, zero-padded), whose transition on the column $(a,b)$ sends a state $K$ to $\mathcal{S}_{a,b} K$, and whose output in the final state $K$ is $K(\mathbf{0}, \mathbf{0})$, namely $\lambda$ at $\lambda\Delta_{\Cc}$, $1$ at $\Delta$, and $0$ elsewhere. This automaton is minimal, and every state is fixed by the zero column, so the computation is stable under zero-padding.
\end{theorem}

The linear structure of the system is equally rigid. The next theorem computes the exact dimension of the $\F_p$-weighted transfer system.

\begin{theorem}
\label{thm:dimension}
Fix a letter $i_0 \in I$. The $\F_p$-linear span of the iterated sections of $\kappa$ admits the basis $\{\kappa\} \cup \{D^{i} : i \in I,\ i \neq i_0\} \cup \{\Delta_{H}, \Delta_{\Cc}\}$, and hence has dimension exactly $n + 2$.
\end{theorem}

Two clarifications regarding finiteness are necessary. The kernel $\kappa$ is an infinite array even over the finite alphabet. Only its columns are finite, with the column of $z$ supported in the box $[\mathbf{0}, z]$. This property explains why the dynamic program of Proposition~\ref{prop:per-step} is finite but increases in size with $S_T$. According to Theorem~\ref{thm:transfer}, the family of sections of $\kappa$ is finite. This finiteness is established by the theorem itself and does not result from the finiteness of the alphabet, which only enters through the number $n$ of carry kernels. The structure of the system is transparent. Besides $\kappa$ and the zero kernel there are $n$ carry kernels, one per letter, and $p + 1$ certificates, namely two indicator kernels, their sum, and the nonzero scalar multiples of one of them. The scalars are the Lucas coefficients $c(e)$ of partial digit blocks, and they are the only place where the prime enters.

\subsection*{Digit lemmas}

The first lemma records the effect of one digit column on membership in $H_p$; it is the base-$p$ form of the doubling rule of the binary case.

\begin{lemma}
\label{lem:digit}
For $x \in \N^{I}$ and $d \in \Dig$, the point $px + d$ lies in $\Cc$ if and only if $\abs{d} \leq p - 1$ and $x \in \Cc$, and $px + d$ lies in $H_p$ if and only if $\abs{d} \geq p$ or $x \in H_p$.
\end{lemma}

\begin{proof}
The digit columns of $px + d$ are $d$ followed by the columns of $x$, since every coordinate of $d$ is smaller than $p$; the Lucas criterion following Theorem~\ref{thm:kummer} applies column by column.
\end{proof}

The second lemma shows that a continuation point at a level divisible by $p$ is divisible by $p$ in every coordinate. It generalizes the parity structure of the binary case.

\begin{lemma}
\label{lem:crossing}
If $x \in \Cc$ and $p \mid \abs{x}$, then $p$ divides every coordinate of $x$, so $x = px'$ with $x' \in \Cc$.
\end{lemma}

\begin{proof}
The bottom column $h_0(x)$ satisfies $\abs{h_0(x)} \equiv \abs{x} \equiv 0 \pmod p$ and $\abs{h_0(x)} \leq p - 1$, hence $h_0(x) = 0$. The membership $x' \in \Cc$ follows from Lemma~\ref{lem:digit}.
\end{proof}

The third lemma isolates the arithmetic of ordering numbers. Its second part is the only source of cancellation modulo $p$ in this section.

\begin{lemma}
\label{lem:rotation}
Let $e \in \N^{I}$. If $\abs{e} \leq p - 1$, then $p \nmid c(e)$. If $e \in \Dig$ and $\abs{e} = p$, then $e$ is mixed, the cyclic group $\Z/p\Z$ acts freely by rotation on the $c(e)$ orderings of a path with increment $e$, and $p \mid c(e)$.
\end{lemma}

\begin{proof}
If $\abs{e} \leq p - 1$, then all factorials in $c(e)$ involve integers smaller than $p$ and are prime to $p$. If $e \in \Dig$ and $\abs{e} = p$, then no single coordinate can carry the whole sum, so at least two coordinates are positive. A word of prime length fixed by a nontrivial rotation is constant, and constant words are excluded, so every orbit of the rotation action has exactly $p$ elements and $p$ divides the number of words; this is also the case of Theorem~\ref{thm:kummer} in which the bottom column carries. At $p = 2$ the rotation is the transposition of the two orders of a mixed pair.
\end{proof}

\subsection*{Localization of counted paths}

The next two lemmas pin down the continuation points of a path near its terminal and near its start. They replace the block bookkeeping of the binary case.

\begin{lemma}
\label{lem:upper}
Let $z \in \N^{I}$, $b \in \Dig$, and $\zeta = pz + b$. If $y \leq \zeta$ and $p$ divides every coordinate of $y$, then $y = pw$ with $w \leq z$. If $y \leq \zeta$, $y \in \Cc$, and $\abs{y} > p\abs{z}$, then $y = pz + r$, where $r = y - pz$ satisfies $0 \leq r \leq b$ and $\abs{r} = \abs{y} - p\abs{z} \leq p - 1$.
\end{lemma}

\begin{proof}
For the first claim, $p w_i \leq p z_i + b_i < p(z_i + 1)$ gives $w_i \leq z_i$. For the second, write $y = pq + r$ with $r = h_0(y) \in \Dig$, so $\abs{r} \leq p - 1$ because $y \in \Cc$, and $q \leq z$ by the first claim applied to $pq \leq y \leq \zeta$. From $p\abs{q} + \abs{r} = \abs{y} > p\abs{z}$ and $\abs{r} \leq p - 1$ we get $\abs{q} \geq \abs{z}$, hence $q = z$, and $r_i = y_i - p z_i \leq b_i$.
\end{proof}

\begin{lemma}
\label{lem:lower}
Let $v \in \N^{I}$, $a \in \Dig$, and $w_0 = pv + a$. If $y \geq w_0$, $y \in \Cc$, and $\abs{y} < p\abs{v} + p$, then $y = pv + r$, where $r = y - pv$ satisfies $a \leq r \in \Dig$ and $\abs{r} = \abs{y} - p\abs{v} \leq p - 1$.
\end{lemma}

\begin{proof}
Write $y = pq + r$ with $r = h_0(y) \in \Dig$ and $\abs{r} \leq p - 1$ as before. From $y_i \geq p v_i$ we get $q \geq v$. From $p\abs{q} \leq \abs{y} < p\abs{v} + p$ we get $\abs{q} \leq \abs{v}$, hence $q = v$, and $r_i = y_i - p v_i \geq a_i$.
\end{proof}

The head lemma shows that a counted path starting at a non-aligned point can reach a point divisible by $p$ only by completing a pure column, and only from a pure starting column.

\begin{lemma}
\label{lem:head}
Let $a \in \Dig$, $a \neq 0$, and let $P$ be a counted path from $w_0 = pv + a$. Suppose $P$ contains a point $y$ at level $p\abs{v} + p$ all of whose coordinates are divisible by $p$. Then $a = k\eps^{j}$ for a unique letter $j$ and a unique integer $1 \leq k \leq p - 1$, the first $p - k$ steps of $P$ are all equal to $\eps^{j}$, the point $y$ equals $p(v + \eps^{j})$, and the points of $P$ below level $p\abs{v} + p$ are exactly $pv + t\eps^{j}$ for $k \leq t \leq p - 1$, each of which lies in $\Cc$ if and only if $v \in \Cc$. Moreover, every non-terminal point of $P$ at level $p\abs{v} + p$ is of the assumed form by Lemma~\ref{lem:crossing}.
\end{lemma}

\begin{proof}
The predecessor $y - \eps^{j}$ of $y$ on $P$, for the appropriate letter $j$, lies at level $p\abs{v} + p - 1$ and is not the terminal of $P$, hence lies in $\Cc$, and Lemma~\ref{lem:lower} gives $y - \eps^{j} = pv + r$ with $a \leq r$ and $\abs{r} = p - 1$. Every coordinate of $y = pv + r + \eps^{j}$ is divisible by $p$, so $r_i + \mathbf{1}[i = j] \in \{0, p\}$ for every $i$, which forces $r = (p-1)\eps^{j}$. From $a \leq r$ and $a \neq 0$ we get $a = k\eps^{j}$ with $1 \leq k \leq p - 1$, and $y = pv + p\eps^{j}$. Every earlier point of $P$ lies in $\Cc$, lies between $w_0$ and $y - \eps^{j}$, and is of the form $pv + r'$ with $k\eps^{j} \leq r' \leq (p-1)\eps^{j}$ by Lemma~\ref{lem:lower}, hence equals $pv + t\eps^{j}$ with $k \leq t \leq p - 1$; conversely all these points occur on $P$, one per level. Each lies in $\Cc$ if and only if $v \in \Cc$ by Lemma~\ref{lem:digit}.
\end{proof}

The contraction lemma identifies aligned counted paths with counted paths at the next scale. It is an exact bijection.

\begin{lemma}
\label{lem:contraction}
For all $u, z \in \N^{I}$, $\varphi(pu, pz) = \varphi(u, z)$ as integers.
\end{lemma}

\begin{proof}
Both sides vanish unless $u \leq z$, and both equal $1$ when $u = z$, so assume $u < z$. Let $P$ be a counted path from $pu$ to $pz$. For $\abs{u} \leq k \leq \abs{z}$ let $y_k$ be the point of $P$ at level $pk$. Each $y_k$ with $k < \abs{z}$ is non-terminal, hence lies in $\Cc$, and equals $p w_k$ with $w_k \in \Cc$ by Lemma~\ref{lem:crossing}; also $y_{\abs{z}} = pz$, and we set $w_{\abs{z}} = z$. The segment of $P$ between $y_k$ and $y_{k+1}$ consists of $p$ unit steps with increment $p(w_{k+1} - w_k)$, so $\abs{w_{k+1} - w_k} = 1$, say $w_{k+1} = w_k + \eps^{j}$, and all $p$ steps of the segment equal $\eps^{j}$. The segment is therefore determined by $w_k$ and $w_{k+1}$, and its interior points $p w_k + t\eps^{j}$ with $1 \leq t \leq p - 1$ lie in $\Cc$ exactly when $w_k \in \Cc$ by Lemma~\ref{lem:digit}, a condition already imposed at $y_k$. The map $P \mapsto (w_{\abs{u}}, \dots, w_{\abs{z}})$ is thus a bijection between the counted paths from $pu$ to $pz$ and the counted paths from $u$ to $z$, the memberships of corresponding points in $\Cc$ and $H_p$ agreeing by Lemma~\ref{lem:digit}.
\end{proof}

\subsection*{The section table of the kernel}

The following proposition computes every section of $\kappa$ in closed form.

\begin{proposition}
\label{prop:table}
Let $a, b \in \Dig$. As kernels,
\[
\mathcal{S}_{a,b}\,\kappa \;=\;
\begin{cases}
\kappa, & a = b = 0, \hfill (\mathrm{i})\\[3pt]
D^{j}, & a = k\eps^{j},\ 1 \leq k \leq p-1,\ b = 0, \hfill (\mathrm{ii})\\[3pt]
\Delta_{H}, & b = a \neq 0,\ \abs{a} \leq p - 1, \hfill (\mathrm{iii})\\[3pt]
\lambda\,\Delta_{\Cc}, & a \leq b,\ a \neq b,\ a \neq 0,\ \abs{a} \leq p-1,\ \abs{b} = p, \hfill (\mathrm{iv})\\[3pt]
\Delta, & b = a,\ \abs{a} \geq p, \hfill (\mathrm{v})\\[3pt]
0, & \text{otherwise}, \hfill (\mathrm{vi})
\end{cases}
\]
where in case $(\mathrm{iv})$ the scalar is $\lambda \equiv c(b - a) \pmod p$, which is nonzero by Lemma~\ref{lem:rotation}. In case $(\mathrm{ii})$ the result does not depend on~$k$.
\end{proposition}

\begin{proof}
Fix $v, z \in \N^{I}$ and put $w_0 = pv + a$ and $\zeta = pz + b$, so that the left side evaluated at $(v,z)$ is $\kappa(w_0, \zeta)$. If $w_0 \not\leq \zeta$, then $\kappa(w_0, \zeta) = 0$ and each right side vanishes as well, the inequality $p v_i + a_i \leq p z_i + b_i$ with $0 \leq a_i, b_i \leq p - 1$ being equivalent to $v \leq z$ together with the stated digit relations in each case. So assume $w_0 \leq \zeta$, which gives $v \leq z$, and split according to $\abs{b}$.

\smallskip
\emph{Aligned terminals, $b = 0$.}
If $a = 0$, Lemma~\ref{lem:contraction} and Lemma~\ref{lem:digit} give $\kappa(pv, pz) = \varphi(v,z)\,\mathbf{1}[pz \in H_p] \bmod p = \kappa(v,z)$, which is case $(\mathrm{i})$.

Let $a \neq 0$. The terminal $\zeta = pz$ lies at the level $p\abs{z}$, which is at least $p\abs{v} + p$ because the start level $p\abs{v} + \abs{a}$ is smaller than $\abs{\zeta}$ for any nonempty path and $\abs{\zeta} \equiv 0 \pmod p$; a trivial path is impossible since $w_0 \neq pz$. Every counted path therefore contains a point at level $p\abs{v} + p$, namely its terminal if $\abs{z} = \abs{v} + 1$ and a non-terminal point otherwise, and in both cases the point has all coordinates divisible by $p$, by assumption in the first case and by Lemma~\ref{lem:crossing} in the second. Lemma~\ref{lem:head} then forces $a = k\eps^{j}$ for a unique letter $j$; when $a$ is mixed there is no counted path and the section is the zero kernel, part of case $(\mathrm{vi})$.

For $a = k\eps^{j}$, Lemma~\ref{lem:head} shows that every counted path begins with the forced pure head from $w_0$ to $p(v + \eps^{j})$, whose points lie in $\Cc$ exactly when $v \in \Cc$, and continues with an arbitrary counted path from $p(v + \eps^{j})$ to $pz$. The head has exactly one realization, independently of $k$. By Lemma~\ref{lem:contraction},
\[
 \kappa(pv + k\eps^{j},\, pz) \;=\; \mathbf{1}[v \in \Cc]\;\varphi(v + \eps^{j}, z)\,\mathbf{1}[pz \in H_p] \bmod p
 \;=\; \mathbf{1}[v \in \Cc]\,\kappa(v + \eps^{j}, z),
\]
where Lemma~\ref{lem:digit} converts the terminal membership. This is $D^{j}(v,z)$, and case $(\mathrm{ii})$ holds, including the degenerate instance $z = v + \eps^{j}$ in which the head itself terminates the path.

\smallskip
\emph{Low terminals, $0 < \abs{b} \leq p - 1$.}
The terminal level $p\abs{z} + \abs{b}$ is not divisible by $p$, and it exceeds $p\abs{z}$. Suppose first that some point $y$ of the path lies in $\Cc$ at level $p\abs{z}$. Then $y = pz$ by Lemmas~\ref{lem:crossing} and~\ref{lem:upper}, so $z \in \Cc$ by Lemma~\ref{lem:digit}. But the terminal satisfies $\zeta \in H_p$ if and only if $z \in H_p$, again by Lemma~\ref{lem:digit} since $\abs{b} \leq p - 1$, so $\zeta \notin H_p$ and $\kappa(w_0, \zeta) = 0$.

Otherwise no point of the path lies in $\Cc$ at level $p\abs{z}$. A point of the path at that level could only be the start lying in $H_p$, every other non-terminal point being in $\Cc$; a start in $H_p$ admits only the trivial path, whose terminal would lie at the start level, below the terminal level. Hence the start level exceeds $p\abs{z}$, which with $v \leq z$ forces $\abs{a} > p(\abs{z} - \abs{v}) \geq 0$. If $\abs{v} < \abs{z}$, then $\abs{a} > p$, so $w_0 \in H_p$ by Lemma~\ref{lem:digit}; a nontrivial path dies at its start, and the trivial path would need $b = a$, impossible since $\abs{b} \leq p - 1 < \abs{a}$. Hence the section vanishes there. If $\abs{v} = \abs{z}$, then $v = z$ and $w_0 = pz + a \leq pz + b$ gives $a \leq b$. All points of a counted path lie strictly between the levels $p\abs{z}$ and $p\abs{z} + p$, so by Lemma~\ref{lem:upper}, together with the domination of $w_0$, every non-terminal point has the form $pz + r$ with $a \leq r \leq b$, and it lies in $\Cc$ exactly when $z \in \Cc$. If $a = b$, the only path is trivial and contributes $\mathbf{1}[\zeta \in H_p] = \mathbf{1}[z \in H_p]$, which over all $(v,z)$ assembles to $\Delta_{H}$, case $(\mathrm{iii})$. If $a < b$, a nontrivial counted path requires $z \in \Cc$ at its start and $\zeta \in H_p$, hence $z \in H_p$, at its terminal, a contradiction, and the section vanishes, completing this range within case $(\mathrm{vi})$.

\smallskip
\emph{Full terminals, $\abs{b} = p$.}
Here $\zeta \in H_p$ always, and the terminal level is $p(\abs{z} + 1)$. No non-terminal point can lie at that level, for such a point would equal $pw$ with $w \leq z$ by Lemmas~\ref{lem:crossing} and~\ref{lem:upper}, whereas its level forces $\abs{w} = \abs{z} + 1$.

Suppose some point of the path lies in $\Cc$ at level $p\abs{z}$. As before that point is $pz$, so $z \in \Cc$, and the path decomposes uniquely into a counted path from $w_0$ to $pz$ with all points in $\Cc$, followed by a final segment of $p$ unit steps from $pz$ to $\zeta$ with increment $b$. The interior points of the final segment are $pz + r$ with $r$ a partial sum of the increment, so $0 < \abs{r} < p$, and they lie in $\Cc$ by Lemma~\ref{lem:digit} since $z \in \Cc$; hence every one of the $c(b)$ orders of the final segment yields a counted path. The number of counted paths through $pz$ is therefore $c(b)$ times the number of initial segments, and $p \mid c(b)$ by Lemma~\ref{lem:rotation} because $b$ is mixed. These contributions vanish modulo~$p$.

Otherwise, disposing of a start in $H_p$ at level $p\abs{z}$ exactly as in the previous case, the start level exceeds $p\abs{z}$, so $\abs{a} > p(\abs{z} - \abs{v}) \geq 0$. If $\abs{v} < \abs{z}$, then $\abs{a} > p$, the start lies in $H_p$, and the trivial path would need $a = b$ with $\abs{a} > p = \abs{b}$, impossible; the contribution is zero. If $\abs{v} = \abs{z}$, then $v = z$ and $a \leq b$. When $\abs{a} \geq p$ the start lies in $H_p$ and only the trivial path survives, which requires $a = b$, hence $\abs{a} = p$, and contributes $\mathbf{1}[\zeta \in H_p] = 1$; over all $(v,z)$ this assembles to $\Delta$, the instance $\abs{a} = p$ of case $(\mathrm{v})$. When $\abs{a} \leq p - 1$, so that $a \neq 0$, every counted path stays strictly between the levels $p\abs{z}$ and $p(\abs{z}+1)$ except for its terminal, so by Lemma~\ref{lem:upper}, together with the domination of $w_0$, its non-terminal points are $pz + r$ with $a \leq r \leq b$ and $\abs{r} \leq p - 1$, all lying in $\Cc$ exactly when $z \in \Cc$, while the terminal lies in $H_p$ unconditionally. The counted paths correspond to the orders of the increment $e = b - a$, of which there are $c(e)$, so the contribution at $v = z$ is $c(e)\,\mathbf{1}[z \in \Cc]$. Over all $(v,z)$ this assembles to $\lambda\Delta_{\Cc}$ with $\lambda \equiv c(e)$, which is nonzero by Lemma~\ref{lem:rotation} since $\abs{e} = p - \abs{a} \leq p - 1$. This is case $(\mathrm{iv})$. For $a = 0$ every contribution passes through $pz$ and vanishes as shown above, part of case $(\mathrm{vi})$.

\smallskip
\emph{High terminals, $\abs{b} > p$.}
Again $\zeta \in H_p$ always. The level $p(\abs{z} + 1)$ is strictly below the terminal level. A path reaching across it would contain a non-terminal point there, equal to $pw$ with $w \leq z$ by Lemmas~\ref{lem:crossing} and~\ref{lem:upper} while $\abs{w} = \abs{z} + 1$, a contradiction. Hence the start level is at least $p(\abs{z} + 1)$, so $\abs{a} \geq p(\abs{z} + 1 - \abs{v}) \geq p$, the start lies in $H_p$, and only trivial paths survive. A trivial path requires $a = b$ and $v = z$ and contributes $\mathbf{1}[\zeta \in H_p] = 1$, assembling to $\Delta$, the instance $\abs{a} > p$ of case $(\mathrm{v})$; all other pairs give zero, completing case $(\mathrm{vi})$.

\smallskip
The six cases are exhaustive and pairwise disjoint, and the proof is complete. All identities except the vanishing arguments through $c(b)$ hold at the level of the integer counts $\varphi$; the reduction modulo $p$ is used only through Lemma~\ref{lem:rotation}.
\end{proof}

\subsection*{Sections of the derived kernels}

The next lemma closes the family by computing every section of the carry and certificate kernels.

\begin{lemma}
\label{lem:closure}
Let $a, b \in \Dig$, $j \in I$, and $\lambda \in \F_p^{\times}$. Then
\begin{align*}
\mathcal{S}_{a,b}\,\Delta &= \begin{cases}\Delta, & a = b,\\ 0, & a \neq b,\end{cases}
&
\mathcal{S}_{a,b}\,(\lambda\Delta_{\Cc}) &= \begin{cases}\lambda\Delta_{\Cc}, & a = b,\ \abs{a} \leq p-1,\\ 0, & \text{otherwise},\end{cases}
\\[4pt]
\mathcal{S}_{a,b}\,\Delta_{H} &= \begin{cases}\Delta_{H}, & a = b,\ \abs{a} \leq p-1,\\ \Delta, & a = b,\ \abs{a} \geq p,\\ 0, & a \neq b,\end{cases}
\end{align*}
and
\[
\mathcal{S}_{a,b}\,D^{j} \;=\;
\begin{cases}
D^{j}, & b = 0 \text{ and } a = k\eps^{j} \text{ for some } 0 \leq k \leq p-1,\\[3pt]
\lambda'\,\Delta_{\Cc}, & \abs{b} = p,\ \abs{a} \leq p - 1,\ a + \eps^{j} \leq b,\\[3pt]
0, & \text{otherwise},
\end{cases}
\]
where $\lambda' \equiv c(b - a - \eps^{j}) \pmod p$ is nonzero. Sections of the zero kernel are zero.
\end{lemma}

\begin{proof}
The equation $pv + a = pz + b$ holds exactly when $a = b$ and $v = z$, which together with Lemma~\ref{lem:digit} for the indicators $\mathbf{1}[pz + b \in H_p]$ and $\mathbf{1}[pv + a \in \Cc]$ gives the three certificate rows.

For the carry kernel, $(\mathcal{S}_{a,b} D^{j})(v,z) = \mathbf{1}[pv + a \in \Cc]\,\kappa(pv + a + \eps^{j},\, pz + b)$, and Lemma~\ref{lem:digit} splits the indicator as $\mathbf{1}[\abs{a} \leq p - 1]\,\mathbf{1}[v \in \Cc]$. If $\abs{a} \geq p$ the section vanishes. Assume $\abs{a} \leq p - 1$.

If $a_j \leq p - 2$, then $pv + a + \eps^{j} = pv + (a + \eps^{j})$ with $a + \eps^{j} \in \Dig$, so the section equals $\mathbf{1}[v \in \Cc]$ times $\mathcal{S}_{a + \eps^{j}, b}\,\kappa$, and Proposition~\ref{prop:table} applies with the column $a + \eps^{j} \neq 0$. Case $(\mathrm{ii})$ there requires $a + \eps^{j}$ pure and $b = 0$, hence $a = k\eps^{j}$ with $0 \leq k \leq p - 2$, and yields $\mathbf{1}[v \in \Cc]\,D^{j} = D^{j}$. Case $(\mathrm{iii})$ yields $\mathbf{1}[v \in \Cc]\,\Delta_{H} = \delta_{v=z}\,\mathbf{1}[z \in H_p]\,\mathbf{1}[z \in \Cc] = 0$, since no point lies in both $H_p$ and $\Cc$. Case $(\mathrm{iv})$, which requires $a + \eps^{j} \leq b$, $a + \eps^{j} \neq b$, $\abs{b} = p$, and $\abs{a} + 1 \leq p - 1$, yields $\mathbf{1}[v \in \Cc]\,\lambda'\Delta_{\Cc} = \lambda'\Delta_{\Cc}$ with $\lambda' \equiv c(e)$ and $e = b - a - \eps^{j} \neq 0$. Case $(\mathrm{v})$, which here occurs exactly when $b = a + \eps^{j}$ and $\abs{a} = p - 1$, hence $\abs{b} = p$ and $e = 0$, yields $\mathbf{1}[v \in \Cc]\,\Delta = \Delta_{\Cc}$, matching $\lambda' = c(0) = 1$. Case $(\mathrm{vi})$ yields zero. The stated second row merges cases $(\mathrm{iv})$ and $(\mathrm{v})$, and $\lambda' \neq 0$ by Lemma~\ref{lem:rotation} since $\abs{e} = p - \abs{a} - 1 \leq p - 1$.

If $a_j = p - 1$, then $\abs{a} \leq p - 1$ forces $a = (p-1)\eps^{j}$, and $pv + a + \eps^{j} = p(v + \eps^{j})$, so the section equals $\mathbf{1}[v \in \Cc]\,\kappa(p(v + \eps^{j}),\, pz + b)$. By cases $(\mathrm{i})$ and $(\mathrm{vi})$ of Proposition~\ref{prop:table} with the zero column, this is $\mathbf{1}[v \in \Cc]\,\kappa(v + \eps^{j}, z) = D^{j}$ when $b = 0$ and zero otherwise, consistent with the first row at $k = p - 1$; the second row cannot occur since $a + \eps^{j} \leq b$ fails for every $b \in \Dig$.
\end{proof}

\subsection*{Proofs of Theorems~\ref{thm:transfer} and~\ref{thm:dimension}}

\begin{proof}[Proof of Theorem~\ref{thm:transfer}]
Closure under every section operator is the content of Proposition~\ref{prop:table} and Lemma~\ref{lem:closure}, and the cardinality count is $1 + 1 + n + 2 + (p - 1) = p + n + 3$ provided the members are pairwise distinct. Each member is reached from $\kappa$ in one step. Indeed $\mathcal{S}_{0,\eps^{0}}\kappa = 0$ and $\mathcal{S}_{\eps^{i},0}\kappa = D^{i}$ and $\mathcal{S}_{\eps^{0},\eps^{0}}\kappa = \Delta_{H}$; for $\Delta$ take $a = b = (p-1)\eps^{0} + \eps^{1}$, which has $\abs{a} = p$; and for $\lambda\Delta_{\Cc}$ with $\lambda = k \in \{1, \dots, p-1\}$ take $a = (p - k)\eps^{1}$ and $b = \eps^{0} + (p-1)\eps^{1}$, so that $b - a = \eps^{0} + (k-1)\eps^{1}$ has $c(b - a) = k$, while $\abs{b} = p$, $0 \neq a \leq b$, and $\abs{a} = p - k \leq p - 1$.

For distinctness, note first the diagonal values. One has $\kappa(x,x) = \Delta_{H}(x,x) = \mathbf{1}[x \in H_p]$, $D^{i}(x,x) = 0$ since $x + \eps^{i} \not\leq x$, $\lambda\Delta_{\Cc}(x,x) = \lambda\,\mathbf{1}[x \in \Cc]$, and $\Delta(x,x) = 1$. Evaluating at $x = \eps^{0} + (p-1)\eps^{1} \in H_p$ and at $x = \mathbf{0} \in \Cc$ therefore separates all members pairwise, except for the pair $\kappa$ and $\Delta_{H}$, which agree on the diagonal, and the pairs formed by two carry kernels or by a carry kernel and the zero kernel, which all vanish there. Off the diagonal every certificate vanishes, while $\kappa((p-1)\eps^{1},\ \eps^{0} + (p-1)\eps^{1}) = 1$, the unique path being a single step from a point of $\Cc$ to a point of $H_p$; this separates $\kappa$ from $\Delta_{H}$. Finally, for $i \in I$ pick a letter $j \neq i$ and set $x = (p-1)\eps^{j}$, so that $x \in \Cc$ and $x + \eps^{i} \in H_p$ by Lemma~\ref{lem:digit}. Then $D^{i}(x,\ x + \eps^{i}) = \kappa(x + \eps^{i},\ x + \eps^{i}) = 1$, while $D^{j'}(x,\ x + \eps^{i}) = 0$ for every $j' \neq i$ because $x + \eps^{j'} \not\leq x + \eps^{i}$. Hence the carry kernels are nonzero and pairwise distinct, and all members of $\mathcal{F}$ are pairwise distinct.

The automaton evaluation proceeds by induction on the number of digit columns of $(v,z)$. Consuming the least significant column $(a,b)$ replaces the current kernel $K$ evaluated at $(pv' + a, pz' + b)$ with $\mathcal{S}_{a,b} K$ evaluated at $(v', z')$. After all columns have been processed, the arguments reduce to $(\mathbf{0}, \mathbf{0})$, and the output $K(\mathbf{0}, \mathbf{0})$ of the final state is the value sought; at the origin $\kappa$, $0$, the $D^{i}$, and $\Delta_{H}$ evaluate to $0$, since $\mathbf{0} \notin H_p$, while $\lambda\Delta_{\Cc}$ evaluates to $\lambda$ and $\Delta$ to $1$. Since the states are pairwise distinct functions, any two of them differ at some pair $(v,z)$, and the runs started at the two states on the digit columns of that pair produce different outputs; no two states are equivalent, and the automaton is minimal by the Myhill--Nerode criterion. Finally, $\mathcal{S}_{0,0}$ fixes every member of $\mathcal{F}$ by Proposition~\ref{prop:table} and Lemma~\ref{lem:closure}, so appending zero columns does not change the final state, which is the stability under zero-padding.
\end{proof}

\begin{proof}[Proof of Theorem~\ref{thm:dimension}]
The span of the iterated sections of $\kappa$ is contained in the span of $\mathcal{F}$, since $\mathcal{F}$ is closed under sections and contains $\kappa$, and it contains $\mathcal{F}$ by Theorem~\ref{thm:transfer}, so the two spans coincide. By Lemma~\ref{lem:kernel-identities}, $D^{i_0} = \kappa - \sum_{i \neq i_0} D^{i} - \Delta_{H}$ and $\Delta = \Delta_{H} + \Delta_{\Cc}$, and every $\lambda\Delta_{\Cc}$ is a scalar multiple of $\Delta_{\Cc}$, so the displayed set spans. For independence, relabel the letters so that $i_0 = 0$ and suppose
\[
 \alpha\,\kappa \;+\; \sum_{i \neq 0} \beta_i\,D^{i} \;+\; \gamma\,\Delta_{H} \;+\; \delta\,\Delta_{\Cc} \;=\; 0 .
\]
Evaluating on the diagonal at $\mathbf{0} \in \Cc$ gives $\delta = 0$. Evaluating at the pair $((p-1)\eps^{1},\ (p-1)\eps^{1} + \eps^{0})$, where $\kappa$ equals $1$ and every $D^{i}$ with $i \neq 0$ vanishes because $(p-1)\eps^{1} + \eps^{i} \not\leq (p-1)\eps^{1} + \eps^{0}$, gives $\alpha = 0$. Evaluating on the diagonal at $\eps^{0} + (p-1)\eps^{1} \in H_p$ gives $\gamma = 0$. Evaluating at $((p-1)\eps^{0},\ (p-1)\eps^{0} + \eps^{i})$ for each $i \neq 0$, where $D^{i}$ equals $1$ and every other $D^{j'}$ vanishes, gives $\beta_i = 0$.
\end{proof}

\subsection*{The binary case}

The classical case of a single fair bit is the specialization $p = 2$, $I = \{0,1\}$, in which the family has seven members.

\begin{corollary}
\label{cor:binary-table}
For $p = 2$ and $I = \{0,1\}$ the transfer system consists of the seven kernels
\begin{align*}
K_0(v,z) &= \kappa(v,z), &
K_1(v,z) &= 0, \\
K_2(v,z) &= \mathbf{1}[v \in \Cc]\, \kappa(v + \eps^1,\, z), &
K_3(v,z) &= \delta_{v=z}\,\mathbf{1}[z \in H_2], \\
K_4(v,z) &= \delta_{v=z}\,\mathbf{1}[v \in \Cc], &
K_5(v,z) &= \mathbf{1}[v \in \Cc]\, \kappa(v + \eps^0,\, z), \\
K_6(v,z) &= \delta_{v=z}, &&
\end{align*}
that is, $K_0 = \kappa$, $K_1 = 0$, $K_2 = D^{1}$, $K_3 = \Delta_{H}$, $K_4 = \Delta_{\Cc}$, $K_5 = D^{0}$, and $K_6 = \Delta$. The sixteen sections $\mathcal{S}_{d_v,d_z}$ act as in Table~\ref{tab:transfer}, and the automaton of Theorem~\ref{thm:transfer} has initial state $K_0$ and accepting set $\{K_4, K_6\}$, the two states with output~$1$.
\end{corollary}

\begin{table}[ht]
\centering
\small
\begin{tabular}{l*{16}{l}}
\toprule
 & \multicolumn{16}{c}{digit column $(d_{v_0}, d_{v_1}, d_{z_0},
   d_{z_1})$,}\\
 & \multicolumn{16}{c}{in lexicographic order $0000, 0001, \dots, 1111$} \\
\cmidrule(lr){2-17}
$K_0$ & 0 & 1 & 1 & 1 & 2 & 3 & 1 & 4 & 5 & 1 & 3 & 4 & 1 & 1 & 1 & 6 \\
$K_1$ & 1 & 1 & 1 & 1 & 1 & 1 & 1 & 1 & 1 & 1 & 1 & 1 & 1 & 1 & 1 & 1 \\
$K_2$ & 2 & 1 & 1 & 4 & 2 & 1 & 1 & 1 & 1 & 1 & 1 & 4 & 1 & 1 & 1 & 1 \\
$K_3$ & 3 & 1 & 1 & 1 & 1 & 3 & 1 & 1 & 1 & 1 & 3 & 1 & 1 & 1 & 1 & 6 \\
$K_4$ & 4 & 1 & 1 & 1 & 1 & 4 & 1 & 1 & 1 & 1 & 4 & 1 & 1 & 1 & 1 & 1 \\
$K_5$ & 5 & 1 & 1 & 4 & 1 & 1 & 1 & 4 & 5 & 1 & 1 & 1 & 1 & 1 & 1 & 1 \\
$K_6$ & 6 & 1 & 1 & 1 & 1 & 6 & 1 & 1 & 1 & 1 & 6 & 1 & 1 & 1 & 1 & 6 \\
\bottomrule
\end{tabular}
\caption{The transfer table for $p = 2$ on the binary alphabet. Entry $j$ in row $K_i$ means
$\mathcal{S}_{d_v,d_z} K_i = K_j$.}
\label{tab:transfer}
\end{table}

\begin{proof}
With $p = 2$ the digit columns are the pairs in $\{0,1\}^2$, the pure nonzero columns are $\eps^0 = (1,0)$ and $\eps^1 = (0,1)$, and the only column of sum $p$ is $(1,1)$, whose multinomial coefficient $c(1,1) = 2$ vanishes modulo~$2$; the only nonzero scalar is $\lambda = 1$, so $\Delta_{\Cc}$ has no proper multiples. The row of $K_0$ is Proposition~\ref{prop:table}. For $d_v = (0,0)$ it yields $K_0$ on the zero column by case $(\mathrm{i})$ and $K_1$ otherwise by case $(\mathrm{vi})$. For $d_v = (1,0)$ it yields $K_5, K_1, K_3, K_4$ for $d_z = (0,0), (0,1), (1,0), (1,1)$, by cases $(\mathrm{ii})$, $(\mathrm{vi})$, $(\mathrm{iii})$, and $(\mathrm{iv})$ with $c(0,1) = 1$; the case $d_v = (0,1)$ yields $K_2, K_3, K_1, K_4$ in the same way, and $d_v = (1,1)$ yields $K_6$ on $d_z = (1,1)$ by case $(\mathrm{v})$ and $K_1$ otherwise. The rows of $K_1, K_3, K_4, K_6$ are the certificate rows of Lemma~\ref{lem:closure}, where $\abs{d} \geq 2$ occurs only for $d = (1,1)$. The row of $K_5 = D^{0}$ is the carry row of Lemma~\ref{lem:closure} with $j = 0$. It returns $K_5$ on the columns $(d_v, d_z)$ with $d_z = (0,0)$ and $d_v \in \{(0,0), (1,0)\}$, returns $K_4$ on the columns with $d_z = (1,1)$ and $d_v + \eps^0 \leq (1,1)$, that is $d_v \in \{(0,0), (0,1)\}$, and returns $K_1$ elsewhere; the row of $K_2 = D^{1}$ is the same with the two letters exchanged. The outputs are $K_4(\mathbf{0},\mathbf{0}) = K_6(\mathbf{0},\mathbf{0}) = 1$ and $0$ for the other five kernels.
\end{proof}

As an initial consistency check, the $K_0$ row confirms part~(c) of Theorem~\ref{thm:BL-fundamental} on $H_2$, and the general form of the check is the same for every prime. For $v = \mathbf{0}$, every input column has $a = 0$, so $\kappa$ either self-loops (on the all-zero column) or moves to the dead state~$0$ by cases $(\mathrm{i})$ and $(\mathrm{vi})$ of Proposition~\ref{prop:table}. Since any $x \neq \mathbf{0}$ contributes at least one nonzero column, such a run ends in the dead state, and $\kappa$ itself evaluates to~$0$ at the origin. Therefore, $\kappa(\mathbf{0}, x) = 0$ for every $x \in H_p$, as required, since $m_E(x) \equiv c(x) \equiv 0 \pmod p$.

The automaton described in Theorem~\ref{thm:transfer} can be represented as the standard five-tuple $(Q, \Sigma, \delta, q_0, \tau)$ of a deterministic finite automaton with output, where the state set is $Q = \mathcal{F}$, the input alphabet is $\Sigma = \Dig \times \Dig$ (representing the $p^{2n}$ digit columns), the transition function $\delta$ is specified by Proposition~\ref{prop:table} and Lemma~\ref{lem:closure}, the initial state is $q_0 = \kappa$, and the output function $\tau(K) = K(\mathbf{0}, \mathbf{0})$ takes values in $\F_p$. For $p = 2$ the output is $0$ or $1$, and the automaton is the ordinary acceptor with accepting set $F = \{K_4, K_6\}$ and transition table Table~\ref{tab:transfer}. The zero kernel is the unique dead state; its row is constant and its output is~$0$.

\begin{remarknum}[The involution of the binary case]
\label{rem:involution}
For $p = 2$, the only step of the proof that is not an exact path bijection is the vanishing of $c(1,1) = 2$ modulo~$2$ in Lemma~\ref{lem:rotation}, which is realized by the fixed-point-free involution swapping the two orders of a terminal mixed block. For a general prime the same role is played by the free action of $\Z/p\Z$ on the $c(b)$ orderings of a terminal full block with mixed column~$b$, and the certificates acquire the scalars $c(e)$ of the partial blocks, which are the Lucas coefficients of the splitting $b = a + e$ of a full column into the part $a$ present at the start and the increment $e$ of the path.
\end{remarknum}

\subsection*{Algorithmic consequences}

The automaton yields the following evaluation cost.

\begin{corollary}
\label{cor:kernel-eval}
$\kappa(v,z)$ can be computed in $O(\log_p \max_{i \in I} \max(v_i, z_i))$ transitions of the automaton, with the digit columns processed in a single pass. Once the transition function has been tabulated, a table of $(p + n + 3)\,p^{2n}$ entries, each transition is a single lookup indexed by a column of $2n$ digits.
\end{corollary}

Combining the automaton with the residue rule yields fast exact fair assignment.

\begin{corollary}
\label{cor:P1-binary}
For a prime modulus $m = p$ and a finite alphabet $I$, Algorithm~\ref{alg:fairsample}, when combined with the residue rule from Lemma~\ref{lem:residue} and rank computation via~\eqref{eq:rank-kernel} and Theorem~\ref{thm:transfer}, achieves exact fair group assignment in $O(\abs{I}\,T \log_p T)$ automaton transitions. For $m = 2$ and $I = \{0,1\}$ this is $O(T \log T)$ bit operations, whereas the dynamic program of Proposition~\ref{prop:per-step} executed modulo~$2$ requires $O(T^2)$ bit operations.
\end{corollary}

\begin{proof}
Exactness follows from Lemma~\ref{lem:residue} and~\eqref{eq:rank-kernel}. The sum in~\eqref{eq:rank-kernel} contains at most $(\abs{I} - 1)\,T$ terms, each corresponding to an automaton evaluation with $O(\log_p T)$ transitions, as established by Corollary~\ref{cor:kernel-eval}, since all coordinates of the branch points and of $S_T$ are at most~$T$. The digit columns of the branch points $S_{t-1} + \eps^{i}$ are maintained incrementally along the walk, with each step altering one coordinate at amortized constant cost. For the binary comparison, the box $[\mathbf{0}, S_T]$ contains $(x_0 + 1)(x_1 + 1) \leq (T/2 + 1)^2$ cells for $x_0 + x_1 = T$, by the arithmetic mean--geometric mean inequality, which is of quadratic order when both letters occur with positive frequency. The dynamic program of Proposition~\ref{prop:per-step} run modulo~$2$, which is also exact by Lemma~\ref{lem:residue}, requires one bit operation per cell. In contrast, using the full integer counts in Algorithm~\ref{alg:fairsample}, where entries reach $\Theta(T)$ bits, increases the cost to $O(T^3)$ bit operations.
\end{proof}

Corollary~\ref{cor:P1-binary} is implemented for $m = 2$ through a self-contained specialization of Algorithm~\ref{alg:fairsample}. The sampling loop (lines~1 to~5 of Algorithm~\ref{alg:fairsample}) remains unchanged, except that for $m = 2$, the stopping test $c(S) \equiv 0 \pmod 2$ in line~4 simplifies, by Theorem~\ref{thm:kummer}, to the constant-time bitwise test $S_0 \mathbin{\&} S_1 \neq 0$. The group-assignment phase (lines~6 to~10), which involves the backward dynamic program for $A(v) = \varphi(v, S)$ over the box $[\mathbf{0}, S]$, the rank $r$, and the block rule $J = \floor{r \cdot m / A(\mathbf{0})} + 1$, is entirely replaced. No count $A(v)$ is computed; instead, the rank is accumulated modulo~$2$ via~\eqref{eq:rank-kernel} as a sum of automaton evaluations $\kappa(S_{t-1} + \eps^0, S_T)$ over the branch points of the observed path, and the residue rule of Lemma~\ref{lem:residue} finalizes the assignment. For a general prime $p$ the same procedure applies with the stopping test of Algorithm~\ref{alg:kummer}, the branch points $S_{t-1} + \eps^{i}$ for $i \prec a_t$, the transition function of Theorem~\ref{thm:transfer}, and the rank accumulated in $\F_p$ from the outputs.

\begin{breakablealgorithm}
\caption{Fair Bit Extraction via the Transfer Automaton (Algorithm~\ref{alg:fairsample} specialized to $m = 2$, $I = \{0,1\}$)}
\label{alg:automaton}
\begin{algorithmic}[1]
\Require Access to i.i.d.\ draws from unknown non-degenerate $\bfpi$
  on $\{0,1\}$
\Ensure Outcome $J \in \{1, 2\}$ with $P(J = j) = 1/2$
\State $S \gets \mathbf{0} \in \N^2$; $\mathcal{B} \gets []$
\Repeat
  \State Draw $a \sim \bfpi$
  \State \textbf{if } $a = 1$ \textbf{ then} append $S + \eps^0$
    to $\mathcal{B}$
    \Comment{branch point of this step; cf.~\eqref{eq:rank-kernel}}
  \State $S \gets S + \eps^a$
\Until{$S_0 \mathbin{\&} S_1 \neq 0$}
  \Comment{$S \in H_2$, i.e.\ $c(S) \equiv 0 \pmod 2$
    (Theorem~\ref{thm:kummer})}
\State $r \gets 0$; \ $D \gets \floor{\log_2 \max(S_0, S_1)}$
\For{$v$ in $\mathcal{B}$}
  \State $q \gets K_0$
  \For{$\alpha = 0, 1, \ldots, D$}
    \State $q \gets \delta\bigl(q,\,
      (\hp(v_0),\, \hp(v_1),\, \hp(S_0),\, \hp(S_1))\bigr)$
      \Comment{row $q$ of Table~\ref{tab:transfer}}
  \EndFor
  \State \textbf{if } $q \in \{K_4, K_6\}$ \textbf{ then}
    $r \gets 1 - r$
    \Comment{$r \gets r \oplus \kappa(v, S)$}
\EndFor
\State \Return $J \gets r + 1$
  \Comment{residue rule, Lemma~\ref{lem:residue}}
\end{algorithmic}
\end{breakablealgorithm}

Correctness and the $O(T \log T)$ complexity follow directly from Corollary~\ref{cor:P1-binary}. The outer loop executes at most $T$ automaton evaluations, each requiring $D + 1 = O(\log T)$ digit columns; the truncation to these columns is exact (Remark~\ref{rem:truncation}). The digit columns of the stored branch points are incrementally maintained along the walk, consistent with the proof of the corollary.

\begin{remarknum}[Digit width of the truncated run]
\label{rem:truncation}
Algorithm~\ref{alg:automaton} processes the $D + 1$ binary columns $\alpha = 0, \ldots, D$ with $D = \floor{\log_2 \max(S_0, S_1)}$, whereas a branch point $v = S_{t-1} + \eps^0$ can occupy one further bit, namely when $v_0 = S_0 + 1 = 2^{D+1}$. This forces $S_0 = 2^{D+1} - 1 \geq S_1$, so the binary expansion of $S_0$ consists of ones only and $S_0 \mathbin{\&} y = y \neq 0$ for every $y \geq 1$; every state $(S_0, y)$ with $y \geq 1$ therefore lies in $H_2$. The walk occupies $(S_0, v_1)$ without stopping just before the final step, so $v_1 = 0$, the walk stops at $S = (S_0, 1)$, and the offending branch point is $v = (2^{D+1}, 0)$, arising at the final step. The truncated run therefore reads only zero columns on the $v$ side and computes $\kappa(\mathbf{0}, S) = c(S) \bmod 2 = 0$, since $S \in H_2$, which agrees with the true value $\kappa(v, S) = 0$, because $v_0 > S_0$ implies $v \not\leq S$. The truncation is exact in every case.
\end{remarknum}

\begin{remarknum}[Beyond prime moduli]
\label{rem:beyond-prime}
Theorem~\ref{thm:transfer} does not extend as stated to composite moduli. For a modulus $m$ with several prime factors, the rank in Lemma~\ref{lem:residue} is taken among the paths avoiding $H_m$ itself, so the automata of the prime factors do not combine, and for a prime power $p^k$ with $k \geq 2$ the cancellation of Lemma~\ref{lem:rotation} is only of first order. The resulting questions are formulated as Open Problem~P1 (Section~\ref{sec:open}).
\end{remarknum}

\section{Numerical Evidence}
\label{sec:numerical}

The following section presents numerical evidence addressing questions that remain unresolved by the theory.

\subsection{Unimodality of the efficiency}
\label{sec:unimodality}

Corollary~\ref{cor:efficiency} demonstrates that the efficiency satisfies $\eta(p) \to 0$ as $p \to 1^+$ and as $p \to \infty$, implying that $\eta$ attains a maximum at some interior $p^\star \in (1,\infty)$. However, it does not establish the uniqueness of this maximum. Figure~\ref{fig:efficiency} displays $\eta(p)$ for two sources. In all computed cases, the curve is unimodal, exhibiting a single interior peak near $p^\star \approx 5$ for $\bfpi = (0.7,0.3)$ and near $p^\star \approx 10$ for the more skewed $\bfpi = (0.9,0.1)$. This pattern suggests, but does not prove, that the interior maximum is always unique.

\begin{figure}[t]
\centering
\begin{tikzpicture}
\begin{axis}[
    width=0.86\textwidth, height=0.5\textwidth,
    xmode=log, log basis x=10,
    xlabel={modulus / scale $p$ (log scale)},
    ylabel={efficiency $\eta(p)=\dfrac{\log p}{E(T)\,\Hsh(\bfpi)}$},
    xmin=1.02, xmax=1000, ymin=0, ymax=0.5,
    xtick={1,2,5,10,20,50,100,200,500,1000},
    xticklabels={1,2,5,10,20,50,100,200,500,1000},
    ytick={0,0.1,0.2,0.3,0.4,0.5},
    grid=both, grid style={gray!20},
    legend pos=north east, legend cell align=left,
    tick label style={font=\small}, label style={font=\small},
]
\addplot[blue!70!black, thick, mark=none] coordinates {
 (1.02,0.0087)(1.05,0.0214)(1.1,0.0416)(1.2,0.0786)(1.35,0.1270)(1.5,0.1685)(1.7,0.2152)(2,0.2711)(2.5176,0.3395)(3.1692,0.3929)(3.9895,0.4265)(5.022,0.4371)(6.3218,0.4272)(7.9579,0.4017)(10.0176,0.3660)(12.6103,0.3254)(15.874,0.2841)(19.9824,0.2452)(25.1542,0.2099)(31.6645,0.1786)(39.8597,0.1514)(50.176,0.1277)(63.1623,0.1074)(79.5096,0.0901)(100.088,0.0753)(125.992,0.0628)(158.601,0.0523)(199.649,0.0434)(251.321,0.0360)(316.367,0.0298)(398.247,0.0246)(501.319,0.0203)(631.068,0.0167)(794.398,0.0138)(1000,0.0113)
};
\addlegendentry{$\bfpi=(0.7,0.3)$}
\addplot[red!70!black, thick, mark=none] coordinates {
 (1.02,0.0060)(1.05,0.0147)(1.1,0.0287)(1.2,0.0547)(1.35,0.0896)(1.5,0.1203)(1.7,0.1562)(2,0.2014)(2.5176,0.2622)(3.1692,0.3176)(3.9895,0.3661)(5.022,0.4062)(6.3218,0.4363)(7.9579,0.4551)(10.0176,0.4613)(12.6103,0.4545)(15.874,0.4352)(19.9824,0.4049)(25.1542,0.3665)(31.6645,0.3237)(39.8597,0.2801)(50.176,0.2388)(63.1623,0.2016)(79.5096,0.1693)(100.088,0.1416)(125.992,0.1181)(158.601,0.0983)(199.649,0.0816)(251.321,0.0676)(316.367,0.0560)(398.247,0.0462)(501.319,0.0381)(631.068,0.0314)(794.398,0.0259)(1000,0.0212)
};
\addlegendentry{$\bfpi=(0.9,0.1)$}
\addplot[blue!70!black, only marks, mark=*, mark size=1.6pt] coordinates {(5.022,0.4371)};
\addplot[red!70!black, only marks, mark=square*, mark size=1.6pt] coordinates {(10.0176,0.4613)};
\node[blue!70!black, font=\scriptsize, anchor=south] at (axis cs:5.022,0.4371) {$p^\star\!\approx\!5$};
\node[red!70!black, font=\scriptsize, anchor=south west] at (axis cs:10.0176,0.4613) {$p^\star\!\approx\!10$};
\end{axis}
\end{tikzpicture}
\caption{The efficiency $\eta(p)$ of the Bernard--Letac sampler as a function of the modulus $p$ (treated as continuous, log scale) for two binary sources. This figure supports Corollary~\ref{cor:efficiency}. Numerically, $\eta$ is unimodal, approaching zero as $p \to 1^+$ (where the target carries negligible information but $E(T) \to C(\bfpi) > 0$) and as $p \to \infty$ (where $E(T) \sim p$), with an interior maximum at $p^\star(\bfpi)$ that shifts to higher values as the source becomes more skewed. Although the uniqueness of the maximum is observed in all computed cases, it remains unproven. The curves are derived from the exact product formula in Theorem~\ref{thm:ET}.}
\label{fig:efficiency}
\end{figure}
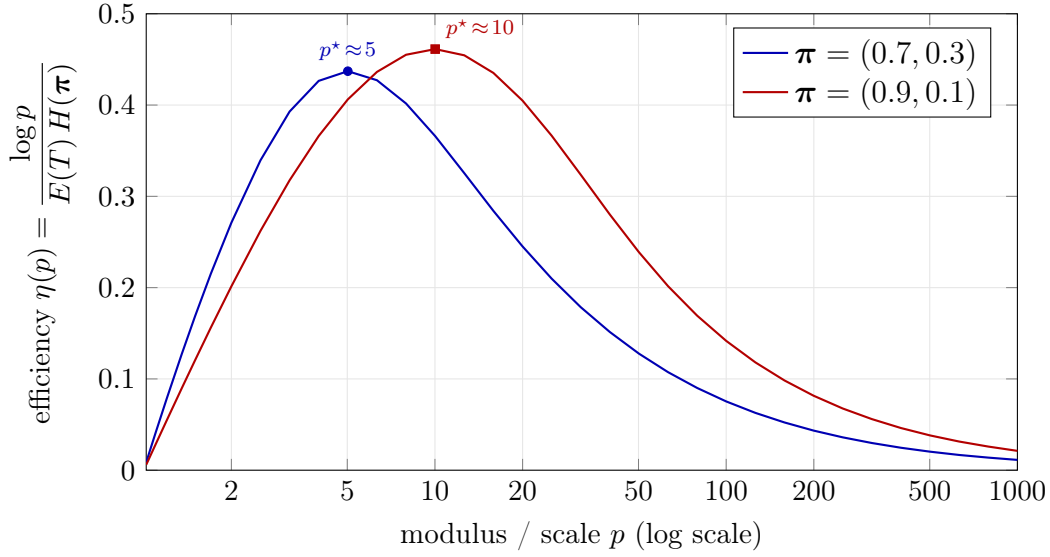

\subsection{A composite-modulus benchmark}
\label{sec:benchmark}

For composite $m$, no closed-form prediction for $E(T)$ is known (Open Problem~P2). Therefore, the following comparison provides genuine empirical evidence rather than a verification of a theorem. The direct Bernard--Letac sampler (Algorithm~\ref{alg:fairsample}) is compared against two Von Neumann (VN) methods for generating a uniform variable in $\{0,\ldots,m-1\}$ from an i.i.d.\ biased
binary source with $P(\text{heads}) = \pi_0$. The metric used is the expected number of biased coin flips per uniform output.

\textbf{VN + Reject}: Von Neumann debiasing is followed by simple rejection sampling. The procedure reads $\lceil\log_2 m\rceil$ fair bits, keeps the values $0,\ldots,m-1$, and rejects and retries whenever the result is at least~$m$.

\textbf{VN + Lumbroso}: Von Neumann debiasing is followed by Lumbroso's Fast Dice Roller~\cite{lumbroso}, which implements the optimal Knuth--Yao DDG tree for uniform distributions. This combination recycles entropy from rejections and is the strongest such method.

\textbf{Bernard--Letac}: This approach uses the direct sampler of Algorithm~\ref{alg:fairsample}, operating on biased flips without intermediate debiasing.

All experiments utilize $2{,}000{,}000$ samples per configuration. The two VN columns report exact Markov-chain values, while the Bernard--Letac column presents the empirical mean. The final column, BL\,/\,Lum, reports the ratio of the Bernard--Letac mean to the VN+Lumbroso value. For prime $m$, the analogous comparison involves only two closed-form quantities (Bernard--Letac via Theorem~\ref{thm:ET} and the VN methods via exact analysis) and is therefore omitted. No closed form for $E(T)$ is known for composite $m$ (Remark~\ref{prop:composite}). The empirical results presented below demonstrate that the Bernard--Letac method retains its advantage over both alternatives.

\begin{center}
\small
\begin{tabular}{cc rrr r}
\toprule
$m$ & $\pi_0$ & VN+Reject & VN+Lumbroso & Bernard--Letac & BL\,/\,Lum \\
\midrule
 4 & 0.50 & 8.00 & 8.00 & 5.86 & 0.73 \\
 4 & 0.70 & 9.52 & 9.52 & 6.30 & 0.66 \\
 6 & 0.50 & 16.00 & 14.67 & 6.65 & 0.45 \\
 6 & 0.70 & 19.05 & 17.46 & 7.18 & 0.41 \\
10 & 0.50 & 25.60 & 18.40 & 7.37 & 0.40 \\
10 & 0.70 & 30.48 & 21.90 & 8.97 & 0.41 \\
\bottomrule
\end{tabular}
\end{center}

\medskip
A remark on the $m = 4$ row is in order for readers cross-reading~\cite{bernardletac}. Remarque~3 of that paper prints $E(T(H_4)) \approx 6{,}2$ for the fair coin, whereas the exact value is $5.86210$, in agreement with the empirical $5.86$ above; the printed figure is a hand estimate and is simply inaccurate. Their companion figure $2\,E(T(H_2)) \approx 6{,}8$, the cost of extracting two fair bits through $H_2$ instead, is correct (exactly $6.80294$).

\medskip
The information-theoretic lower bound $\log m / \Hsh(\pi_0)$ of Theorem~\ref{thm:lower-bound} is substantially lower than the performance of all methods, indicating that none approaches optimality in absolute terms. The advantage of the Bernard--Letac method is structural. By omitting the debiasing stage, it extracts more useful information per biased flip.

\section{Open Problems and Future Directions}
\label{sec:open}

Several open problems arise from the preceding analysis.

\begin{enumerate}[label=\textbf{P\arabic*.}, leftmargin=*, itemsep=6pt]

 \item \textbf{Transfer systems beyond prime moduli.} Theorem~\ref{thm:transfer} settles the prime case of the transfer question. For every prime $p$ and finite alphabet $I$, the kernel $\varphi(v,z) \bmod p$ is computed by a minimal automaton with exactly $p + |I| + 3$ states, the $\F_p$-weighted transfer system spanned by its states has dimension exactly $|I| + 2$ (Theorem~\ref{thm:dimension}), and the group assignment of Algorithm~\ref{alg:fairsample} costs $O(|I|\,T \log_p T)$ automaton transitions (Corollary~\ref{cor:P1-binary}). Composite moduli remain open. Since the rank of Lemma~\ref{lem:residue} is taken among the paths avoiding $H_m$ itself, the automata of the prime factors of $m$ do not combine, and the assignment for composite $m$ still costs the dynamic program of Proposition~\ref{prop:per-step}. Two questions arise. For a prime power $m = p^k$ with $k \geq 2$, the stopping set $H_{p^k}$ is itself $p$-automatic, because the carries of Theorem~\ref{thm:kummer} are bounded by $|I| - 1$ and a finite automaton reading the digit columns least significant first can track their total up to~$k$; but the cancellation of Lemma~\ref{lem:rotation} is only of first order, since $c(b)$ is divisible by $p$ and in general not by $p^2$ for a mixed column $b$ of sum $p$, so the certificate structure of Theorem~\ref{thm:transfer} must change. Does the kernel modulo $p^k$ admit a finite transfer system in base $p$, and of what size? The case $p = 2$, $k = 2$ on the binary alphabet is the first test. For a modulus with two distinct prime factors $p$ and $q$, the stopping set $H_p \cap H_q$ is not semilinear. On the binary alphabet the integers $x$ with $(x,x) \notin H_p \cap H_q$ are those whose base-$p$ digits are all at most $(p-1)/2$ or whose base-$q$ digits are all at most $(q-1)/2$; this set is infinite, containing the powers of any odd prime among $p$ and $q$, and has density zero, whereas an infinite eventually periodic set has positive density. By the theorem of Cobham and Semenov~\cite{cobham,semenov}, a subset of $\N^{I}$ that is automatic in two multiplicatively independent bases is semilinear, so $H_p \cap H_q$ cannot be read by automata in both bases; whether it, or the kernel modulo $pq$, is automatic in a single base is unknown, and no exact assignment for composite $m$ running in time $O(T \log^{O(1)} T)$ is known. The interaction with Kummer's multi-carry criterion is also the obstacle in~P2.

\item \textbf{Closed-form $E(T)$ for composite $m$.} Remark~\ref{prop:composite} explains why the derivation of Theorem~\ref{thm:ET} does not extend to composite~$m$. The question remains whether a formula for $E(T)$ exists in terms of the R\'enyi entropies at the prime-power orders $p_j^{a_j}$. The inclusion-exclusion structure of $H_m = \bigcap_j H_{p_j^{a_j}}$ suggests that a M\"{o}bius-type inversion could be relevant.

\item \textbf{Entropy order and the optimal modulus.} Given a fixed source $\bfpi$ and a required number of outcomes $m$, it is of interest to determine which decomposition of $m$ into prime powers minimizes $E(T)$. Since prime moduli allow for closed-form analysis, sampling instead with the modulus $p$, the smallest prime satisfying $p \geq m$, may often be more efficient in practice than using composite $m$, even if this requires a slightly larger outcome alphabet.

\item \textbf{Adaptive stopping boundaries.} The set $H_m$ is predetermined. It is an open question whether a data-adaptive stopping boundary can be designed that, after estimating $\bfpi$ online, dynamically selects a subset of $H_m$ to reduce $E(T)$ while maintaining fairness.

\item \textbf{$p$-adic limits and entropy.} The $p$-adic limit $L(x)$ of Algorithm~\ref{alg:padic} is a $\Zp$-valued function on $\Mon$. It remains to be determined whether $L(x)$ can be interpreted in terms of the entropy of a distribution associated with $x$. The product formula $L(x)\,L(y) = \sum_z L(z)$ of Th\'eor\`eme~10 of~\cite{bernardletac}, in which $z$ ranges over matrices with margins $p^k x$ and $p^k y$ across all scales $k \geq 0$, resembles a convolution, and the connection to $p$-adic $L$-functions and $p$-adic integration may yield further insights.

\item \textbf{Sources of infinite Shannon entropy.} On a countable alphabet every R\'{e}nyi entropy of order $p > 1$ is finite, since $0 < \sum_i \pi_i^p \leq (\max_i \pi_i)^{p-1} < 1$; only the orders $\alpha \leq 1$ can diverge. Theorem~\ref{thm:ET} and its R\'{e}nyi reading therefore carry over to countable $I$ unchanged, and $E(T) > p$ strictly in every case by Proposition~\ref{prop:ET-monotone}(a). The limit constant $C(\bfpi)$ of Theorem~\ref{thm:shannon-limit} likewise remains finite. The convergence of the integral defining $C(\bfpi)$ rests only on the behavior of $F$ near $s = 0$ and on the doubly exponential decay of $\varphi$, neither of which involves $\Hsh(\bfpi)$. The open question concerns the boundary. For a source with $\Hsh(\bfpi) = \infty$ but $H_\alpha(\bfpi) < \infty$ for every $\alpha > 1$, what replaces the lower bound $(\log p)/\Hsh(\bfpi)$ of Theorem~\ref{thm:lower-bound}, which becomes vacuous?

\item \textbf{Multivariate Kummer carries and zero-one matrix counts.} A corollary of \S5 of~\cite{bernardletac} provides an upper bound on the number of zero-one matrices with specified margins. Refining this bound using the carry structure of Algorithm~\ref{alg:kummer} remains an open problem, which is related to the longstanding challenge identified by Ryser~\cite{ryser}.

\end{enumerate}

\bigskip
\paragraph*{Acknowledgements.}
This work received no specific grant from any funding agency, commercial or
not-for-profit sectors.

\paragraph*{Competing interests.}
The author declares none.

\paragraph*{Data availability statement.}
No external data were used in this work. The numerical results of Section~\ref{sec:numerical} were produced by the author's implementation of the algorithms described in the paper.


\end{document}